\documentclass[
  reprint,              % two-column PRL style
  amsmath,amssymb,
  aps,
  prl,
  superscriptaddress
]{revtex4-2}

\usepackage{graphicx}
\usepackage{dcolumn}
\usepackage{bm}
\usepackage{hyperref}   % ok for submission; PR may strip links in production
\usepackage{physics}    % optional; remove if you prefer plain LaTeX
\IfFileExists{siunitx.sty}{\usepackage{siunitx}}{}    % optional; units
\usepackage{xcolor}     % for named colors like orange

\usepackage{amsmath,amsthm,verbatim,amssymb,amsfonts,amscd, graphicx, enumitem}

\definecolor{DarkMagenta}{RGB}{139,0,139}
\newcommand{\CHANGE}[1]{{\color{DarkMagenta}#1}}

\newtheorem{theorem}{Theorem}

\newtheorem{proposition}{Proposition}
\newtheorem{definition}{Definition}

\newcommand{\E}{\mathbb{E}}
\newcommand{\V}{{\rm Var}}
\newcommand{\sgn}{{\rm sgn}}

\newcommand{\R}{\mathbb{R}}
\newcommand{\LR}{\Lambda}
\renewcommand{\grad}{\nabla}
\newcommand{\Hess}{\nabla^2}

\newcommand{\YR}[1]{{\color{orange}[YR:#1]}}

\begin{document}

% Title
\title{Thermodynamic Irreversibility of Training Algorithms}

%\author{...}

% Authors and affiliations
\author{Liu Ziyin*}
\affiliation{Massachusetts Institute of Technology}
\affiliation{NTT Research}

\author{Yuanjie Ren*}
\affiliation{Massachusetts Institute of Technology}

\author{Adam Levine}
\affiliation{Massachusetts Institute of Technology}

\author{Isaac Chuang}
\affiliation{Massachusetts Institute of Technology}

% If you have multiple affiliations, use \affiliation multiple times.
% If you want an email footnote:
% \email{corresponding.author@university.edu}

% Date (optional)
\date{\today}
% Abstract (PRL abstracts are short and punchy)
\begin{abstract}
    The training algorithms for AI systems all introduce far-from-equilibrium dynamical processes, and understanding the irreversibility of these algorithms is a fundamental step towards understanding the learning dynamics of modern AI systems. In this work, we establish a general framework for defining and analyzing the irreversibility of training algorithms. We show that four different ways to characterize the irreversibility of dynamical processes are equivalent to leading order in the step size $\eta$: numerical backward error $\phi_{\rm DE}$, time-renormalized correction $\phi_{\rm TR}$, microscopic time reversal asymmetry  $\phi_{\rm TA}$, and the (regularized) stochastic-thermodynamic entropy production $\phi_{\rm ST}$. The irreversibility gives rise to a time-reversal-symmetry-breaking emergent force that generically breaks non-isometric continuous reparametrization symmetries, preserves orthogonal symmetries, and leads to a universal preference for those learning trajectories that minimize the entropy production rate. 
\end{abstract}

% PRL requires PACS no longer; keywords usually not needed.
\maketitle

% -------------------- Main text --------------------
% PRL letters are short; aim for ~4 pages in reprint format.

\paragraph{Introduction.} Understanding the training dynamics of contemporary AI models has been an important task with both engineering and scientific values. The learning process of neural networks (NN) can be regarded as an interacting particle system evolving as a nonequilibrium process \cite{mei2019mean, halverson2021neural, rotskoff2022trainability, coppola2025renormalization, xie2020diffusion}, and it has been shown to exhibit many signatures of irreversible mechanisms. For example, a subfield of AI is called ``machine unlearning," which is based on the widely held observation that learning a data point is quite different from ``forgetting" a data point, a process that can be, arguably, regarded as the reverse process of learning. Having a formalism to analyze the irreversibility of the NN training will enable the application of techniques and ideas from high-energy and condensed matter physics to the analysis of the training dynamics.

In nonequilibrium physics, the most important quantity is, to no one's surprise, time irreversibility, which is often quantified by the entropy production of the system \cite{prigogine1973theory, seifert2008stochastic, o2022time}. Having a mathematical formalism to study the entropy production during the training dynamics of NNs will enable us to leverage tools from, e.g., stochastic thermodynamics (ST), to study the learning dynamics of modern neural networks. 
In deep learning, however, the concept and meaning of entropy production and time reversibility have been unclear, even if various phenomenologies have suggested that irreversibility must play a role in determining many aspects of NN training \footnote{For example, see the discussion in \cite{ziyin2025neural}}. In this work, we build a formal connection between different notions of irreversibilities that are applicable to understanding the training algorithms, ranging from conventional tools from numerical analysis to ideas in theoretical physics.

\paragraph{Setting: Training dynamics of deep learning.} Consider a generic discrete-time dynamics:
\begin{equation}\label{eq: update rule}
    \Delta \theta_t =  -\eta U(\theta_t),
\end{equation}
where $\Delta$ is the time difference operator: $\Delta A_t = A_{t+1} - A_t$, and $\eta$ is the time step. In a machine learning (ML) context, $\eta$ is also called the learning rate. The minus sign is chosen for convention. Because \eqref{eq: update rule} can be seen as a time evolution, the sign of $\eta$ is of (often neglected) crucial importance: ${\rm sgn}(\eta)$ is the arrow of time for the system. When $\eta >0$ ($<0$), the system is running forward (backward) in time.

In deep learning, we are given data distribution $P(x,y)$ (the training set), and want to train a model $f_\theta(x)$ such that $\ell$ is a loss function that minimizes some form of distance or divergence between $f(x)$ and $y$:
\begin{equation}
    L(\theta) = \E[\ell(f_\theta(x), y)]
\end{equation}
where, throughout this work, $\E =\E_{x,y\sim P(x,y)}$ is the empirical average. A canonical choice of $U$ to minimize this $L$ is the SGD algorithm:
\begin{equation}
    U = \nabla L.
\end{equation}
It is also possible and common for $U$ to be random and dependent on the sampling of the data: $U=U_{x,y}$. In that case, our results will also be correct for these stochastic updates and their expectations. To keep the discussion minimal, we will only comment on the difference between stochastic update rules and deterministic ones when necessary.

While it is helpful to think of $U$ as a gradient, our theoretical results are more general and apply to generic discrete-time evolutions where the vector field possesses a symmetric Jacobian $J_U=J_U^T$ where $T$ is transpose. The symmetric Jacobian is naturally satisfied when the update rule is conservative (e.g., $U = \nabla L$, where $J_U = \nabla^2 L$ is the symmetric Hessian). Therefore, standard adaptive learning algorithms such as Adam~\cite{journals/corr/KingmaB14_adam} and RMSProp~\cite{Tieleman2012_rmsprop} can also be modeled under this framework when their effective update field, possibly after augmenting the state, satisfies this condition. %This condition guarantees that the continuous-time backward error force is strictly conservative and can be integrated into a scalar potential. Non-conservative fields with an antisymmetric rotational component can be incorporated into this framework by isolating their conservative projections. Therefore, algorithms with non-symmetric Jacobians, such as Ref.~\cite{journals/corr/KingmaB14_adam} and Ref.~\cite{Tieleman2012_rmsprop}, can be modeled under this framework.
Our result also applies to many interesting systems in soft matter physics and ecology, where the evolution dynamics is commonly modeled to be discrete-time \cite{may1976simple}. Therefore, even though our primary discussions will be about learning algorithms of AI models, our results are more broadly applicable to nonequilibrium systems.

\paragraph{Main results.} We start with the conventional backward error analysis from numerical analysis \cite{hairer2006geometric}, and then move on to three types of physical irreversibilities. 
All the derivations are proved as formal theorems in the supplemental material. In the main text, we only name the key steps and focus on their implications.

\paragraph{Discretization Error: $\phi_{\rm DE}$.} The old problem with the discrete dynamics is that it is not always clear what the dynamics is actually minimizing, even if $U$ can be regarded as coming from an energy field: $U= \nabla E$, running $U$ forward in time is not guaranteed to minimize $E$. Therefore, discrete-time dynamics has some level of intrinsic ambiguity. The most well-known example is perhaps the discrete-time chaos \cite{may1976simple}. In contrast, the continuous-time dynamics is easier to understand. Consider the continuous time limit of \eqref{eq: update rule} with ($U= \nabla E$):
\begin{equation}\label{eq: continuous update}
    \dot \theta = -\eta \nabla E.
\end{equation}
This system is guaranteed to minimize $E$ (though not necessarily to a local or global minimum) because $\dot{E} = -\eta \|g\|^2 \leq 0$, where $g = \nabla E$. Also, because most tools of physics are developed for continuous-time evolution, it becomes possible to apply knowledge of physics to these problems. For example, Eq.~\eqref{eq: continuous update} can be interpreted as an overdamped dynamics, and is thus dissipative.

Thus, an important question is whether and how one can construct an effective continuous-time dynamics $\hat{U}$ that approximates the original \eqref{eq: update rule} even if the step size is not infinitesimally small:
\begin{equation}
    \dot{\theta}=-\eta \hat U_\eta(\theta),\qquad \hat U_\eta(\theta)=U(\theta)+\eta F(\theta)+O(\eta^2).
\end{equation}
Matching the time-one flow of this modified dynamics to one discrete update gives
\begin{equation}
   F(\theta)=\frac{1}{2}J_U(\theta)U(\theta).
\end{equation}
Because $J_U U=\nabla(\|U\|^2/2)$, the leading scalar discretization-error potential satisfies
\begin{equation}
   \nabla \phi_{\rm DE}(\theta)=\frac{\eta}{2}J_U(\theta)U(\theta),\qquad \phi_{\rm DE}(\theta)=\frac{\eta}{4}\|U(\theta)\|^2.
\end{equation}
Therefore, $\phi_{\rm DE}$ is the squared norm of the update rule up to the prefactor $\eta/4$. This term admits the direct interpretation as the fluctuation of the learning dynamics. We emphasize that this insight is highly physical: stochastic dynamics has the tendency to move to a place with low fluctuation. This fact is evident in the common phenomenon that a temperature gradient causes suspended particles to drift -- typically toward the cooler region (thermophoresis).
%, and in the mundane observation that cool air is denser than hot air and so accumulates at the bottom of a room. 
% not really - given the ideal gas law, this is deterministic and fully explained by macroscopic law; in the thermodynamic limit all the randomness has averaged away}
If $U$ is the gradient of a loss function $L$, one can directly write down the effective free energy that the system is minimizing:
\begin{equation}\label{eq: effective free energy}
    F_{\rm eff} = L + \frac{\eta}{4}\|\nabla L\|^2.
\end{equation}
Thus, optimizing $L$ with Euler-discretized gradient descent is, to leading order, equivalent to following gradient flow (GF) on $F_{\rm eff}$. This analysis has been known as backward error analysis, and it answers a question closely related to ``reversibility": if we solve the problem $\theta$ with an approximate algorithm $U$, then the problem actually being solved is $F_{\rm eff}$ \cite{hairer2006geometric, smith2021origin, barrett2020implicit}.

It is tempting to call $\eta$ a ``temperature" because it is proportional to the level of fluctuation in the learning dynamics. However, the training dynamics is inherently a nonequilibrium process and it maybe quite misleading to think of $\eta$ as temperature and lead to many conceptual paradoxes: (1) $\eta$ is both proportional to the drift and to the diffusion; (2) it is a common phenomena that a higher $\eta$ leads, surprisingly, to a lower fluctuation both in the parameters \cite{ziyin2025noise} and gradients, as we will demonstrate below.

\paragraph{Time Renormalization: $\phi_{\mathrm{TR}}$.} 
The {Renormalization Group (RG)} is a coarse-graining framework used to extract large-scale, continuum behaviors from discrete systems \cite{wilson1971renormalization}. This process is inherently {irreversible}; as we ``zoom out,'' microscopic information is lost. This loss is quantified by {c-theorems}, which identify monotonic variables that act as an ``entropy'' metric for the coarse-graining flow.
In discrete learning dynamics, we renormalize time to isolate the long-term behaviors most relevant to optimization. 
We apply RG to the time axis of the learning rule: by repeatedly merging
neighboring update steps, we discard short-timescale transients and isolate
the long-term dynamics, the part most relevant to learning.
 
This is natural for discrete-time dynamics, where the step size $\eta$ defines
a time lattice. Each round of renormalization halves the lattice constant: set
$U_0 := U$, $\eta_0 = \eta$, $\eta_k = \eta_{k-1}/2$, and define $U_{k+1}$ to
be the rule satisfying
\begin{equation}
    \theta_2 - \theta_0 = -\eta_k U_k(\theta_0) + O\!\left(\eta_k^3\right)
    \label{eq:consistency}
\end{equation}
when run for two steps at $\eta_{k+1}$; in this way two fine steps of $U_{k+1}$ reproduce one coarse step of $U_k$ to leading order. Expanding in $\eta_{k-1}$ yields the RG flow equation
\begin{equation}
    U_k = U_{k-1} + \nabla\!\left(\frac{\eta_{k-1}}{8}\|U_{k-1}\|^2\right)+O(\eta_{k-1}^2),
    \label{eq:rg-flow}
\end{equation}
with accumulated coupling strength
\begin{equation}
    g(k) = \frac{1}{4}\sum_{i=1}^{k} 2^{-i} = \frac{1}{4}\!\left(1 - 2^{-k}\right),
    \label{eq:coupling}
\end{equation}
so that, to leading order, the cumulative potential after $k$ refinements is
\begin{equation}
    \phi_{\mathrm{TR}}^{(k)}(\theta)=\eta\,g(k)\|U(\theta)\|^2 + O(\eta^2).
\end{equation}
%As $k \to \infty$, $g(k) \to 1/4$: the rule reaches an RG fixed point, and the
%coarse-grained dynamics converge to~$\phi_{\mathrm{DE}}$.
Apparently, $k\to-\infty$ corresponds to a coarse-graining process, and the coupling strength diverges in that limit, so the emergent term is indeed relevant. As $k \to \infty$ (the continuous limit), which corresponds to a fine-graining process, $g(k)$ converges to {1/4}. This confirms that the time-renormalized dynamics align with the discretization error $\phi_{\text{DE}}$ derived via backward error analysis to leading order.

\begin{comment}
\paragraph{Microscopic Time Reversal Asymmetry: $\phi_{\rm TA}$.} More directly, one can take the obvious route for studying whether the dynamics is microscopically reversible in time. As discussed above, $\eta$ is the arrow of time, and so the algorithm will have the microscopic time reversal antisymmetry \footnote{Namely, the dynamics changes sign if the arrow of time $\eta$ changes sign.} if reversing the arrow of time will make the particle return to its original place. Operationally, this means that taking one step forward in time can be exactly canceled by taking one step back in time. Namely, the difference between the two time-reversed updates characterizes the time reversal antisymmetry of the training algorithm:
\begin{equation}
    \phi_{\rm TA} := \|\theta_t - \tilde{\theta}_t\|^2.
\end{equation}
where $\tilde{\theta}_t= \theta_t - \eta U(\theta_t) +\eta  U(\theta_{t+1})$. When this equation holds, the microscopic dynamics has a manifest time-reversal antisymmetry. A direct computation shows that 
\begin{equation}
    \phi_{\rm TA} = \eta \|U\|^2 + O(\eta^2).
\end{equation}
Therefore, the algorithm $U$ generically breaks the antisymmetry, which only emerges in the continuous-time limit. This irreversibility can also be formalized as an emergent heat of the learning algorithm, which we discuss in the next part.
\end{comment}

\paragraph{Microscopic Time Reversal Asymmetry: $\phi_{\mathrm{TA}}$.}
A natural way to probe the irreversibility of a discrete update rule is to ask
whether the dynamics is microscopically time-reversible. Since the step size
$\eta$ plays the role of the arrow of time, reversing time corresponds to
negating $\eta$. The update rule $U$ is microscopically time-reversible if
applying one forward step followed by one backward step returns the system to
its starting point:
\begin{equation}
    \theta_{t+1} = \theta_t - \eta\, U(\theta_t)
    \quad\Longrightarrow\quad
    \theta_{t+1} + \eta\, U(\theta_{t+1}) = \theta_t.
    \label{eq:reversibility-condition}
\end{equation}
Define $\tilde{\theta}_t$ as the result of taking one step \emph{forward} from
$\theta_t$ and then one step \emph{backward} from $\theta_{t+1}$:
\begin{equation}
    \tilde{\theta}_t := \theta_t - \eta U(\theta_t) + \eta U(\theta_{t+1}).
    \label{eq:theta-tilde}
\end{equation}
Perfect reversibility would require $\tilde{\theta}_t = \theta_t$. The difference therefore measures the degree of time-reversal antisymmetry of the
algorithm, with a conventional factor of $2\eta$:
\begin{equation}
    2\eta  \nabla \phi_{\rm TA} := \theta_t - \tilde{\theta}_t.
    \label{eq:phi-ta-def}
\end{equation}
Substituting~\eqref{eq:theta-tilde} and expanding $U(\theta_{t+1}) =
U(\theta_t) - \eta\,\nabla U(\theta_t)\,U(\theta_t) + O(\eta^2)$ gives
\begin{align}
    \phi_{\mathrm{TA}}
        &=   \frac{\eta}{4} \| U\|^2  +   O(\eta^2)
\end{align}
where the second equality uses $\nabla \phi_{\rm DE}=(\eta/2)J_UU$. This shows that a
generic discrete-time algorithm $U$ breaks time-reversal symmetry at finite step size; the discrepancy vanishes only in the continuous-time limit
$\eta \to 0$. This irreversibility can equivalently be characterized as an
``emergent heat'' of the learning algorithm, discussed in the next section. Interestingly, by taking one step forward and one step backward, the algorithm minimizes the emergent potential $\phi_{\mathrm{TA}}$.

%In conventional machine learning theory, this type of analysis has been done using the method of backward error analysis for stochastic gradient descent \cite[see Section 3]{smith2021origin}. 

%To relate discretization error to the thermodynamic concept of reversibility, we observe that the discrete dynamics defined by $U$ are inherently irreversible due to discretization artifacts; applying the backward rule defined by $-U$ to the updated parameter does not precisely recover the initial state.}

% It is also worthwhile to comment on this derivation from the perspective of conventional numerical analysis. In numerical analysis, it is well known that the numerical integrators tend to solve an alternative ``modified" problem that is perturbatively away from the original. In our theory, the entropic loss in Eq.~\eqref{eq: entropic loss} is exactly the modified problem being solved by the SGD algorithm (and $\eta$ is the perturbative parameter), and Eq.~\eqref{eq: entropic loss} can be seen as the modified energy functional of the continuous gradient flow when integrated with Euler discretization. It is an established wisdom in numerical analysis that it is these backward errors (the entropy term) that dominate the long-term behavior of the integrator, especially when there are symmetries in the dynamical variables \cite{hairer2006geometric}. Thus, it is no surprise that these integrators select very special solutions from a large degenerate manifold of solutions. [todo]

\paragraph{Stochastic Thermodynamics: $\phi_{\rm ST}$.} A key advantage of the ST framework is that it has a direct meaning when applied to physical systems such as chemical networks and also allows for the study of generic stochastic processes that may not have a direct physical meaning \cite{seifert2008stochastic, ziyin2023universal, liu2019thermodynamic}. Notably, Ref.~\cite{goldt2017stochastic} studies the stochastic thermodynamics of the Hebbian learning rule for a neural network, but it is unclear how to extend the theory to standard stochastic gradient training in the deep learning practice. Thus, despite this flexibility of ST, its application to deep learning is not easy because the stochastic process induced by the training algorithm often has a measure zero. This is because when the data set size is finite, every, say, SGD step can only take the parameters to a finite set of possible locations, and the reverse direction often has a zero probability, and this corresponds to a fully absolutely irreversible process \cite{murashita2014nonequilibrium}, and the standard ST does not offer much insight.

To achieve this, we define a virtual temperature $\sigma^2$ and virtual spread $\tau$ such that the learning dynamics follow a discrete-time Langevin process with noise scale $\sigma$. The irreversibility of stochastic thermodynamics is characterized by the trajectory-averaged entropy production (EP) \cite{seifert2008stochastic}: $\Sigma  = \langle \log (P([\theta])/P([\theta]^\dagger))\rangle$, where $P([\theta])$ denotes the probability of the forward trajectory and $P([\theta]^\dagger)$ that of the backward trajectory. At thermal equilibrium, physical systems must have that the forward trajectory and the backward trajectory have the same probability, and so the EP is zero at equilibrium, and its magnitude characterizes the degree of irreversibility of the system. It is known that the EP can be decomposed into a sum of two terms: $\Sigma = \Sigma_{\rm sys} + \Sigma_{\rm heat}$, where $\Sigma_{\rm sys}$ is the change in the information entropy of the system, and $\Sigma_{\rm heat}$ is the entropy flow into the bath.

%We focus on the entropy production in the medium, $\Delta S_{\text{med}}$, which captures the heat dissipation. We derive the effective potential that governs this entropy production in the limit of small fluctuations.

Consider the discrete Langevin dynamics $\theta_{k+1} = \theta_k - \eta U(\theta_k) + \sigma \varepsilon_k$, where $\varepsilon_k \sim \mathcal{N}(0, I)$ and $\sigma^2 \propto \eta$. Also, let the system be initialized in a local Gaussian state $\theta_k \sim \mathcal{N}(\mu, \tau I)$. Because neither $\sigma$ or $\tau$ are ``physical," we will take the zero $\sigma$ and zero $\tau$ limit eventually. Introducing these virtual randomnesses is essentially the same as taking an ultraviolet cutoff to prevent the divergence of the entropy production for absolutely irreversible processes.  %In the limit $\tau \to 0$ (approximating a point mass trajectory), the expected medium entropy production $\mathbb{E}[\Delta S_{\text{med}}]$ is governed by the thermodynamic potential $\phi_{\rm ST} = \frac{\eta}{4} \|U\|^2$:

One obtains 
\begin{align}
 \lim_{\tau \to 0}\left(\lim_{\sigma^2 \to 0} \sigma^2 \Sigma\right) = &\lim_{\tau \to 0}\left(\lim_{\sigma^2 \to 0} \sigma^2 \Sigma_{\rm heat}\right)\nonumber\\
   :=& {8\eta} \phi_{\rm ST}(\mu)  + O(\eta^3).
\end{align}
Note that this is only the EP for a single sampling of a data point, and so $\phi_{\rm ST}$ is only the trajectory-wise EP. When there are multiple possible samplings over data distribution, the expectation of $\phi_{\rm ST}$ over the dataset gives the ensemble entropy production: $\E[\phi_{\rm ST}]$.

This result reveals that the training dynamics follow a ``Principle of Minimal Dissipation." The term $\phi_{\rm ST}$ represents the thermodynamic cost or ``heat" generated by the algorithm's discrete fluctuations (due to either stochasticity or discretization error). Note that this is not the same as saying that the system seeks the state with the lowest entropy, but that the system dynamics (namely, the learning trajectories) seeks those with the lowest entropy production rate. A direct consequence of the theory is that by minimizing this potential, the optimizer does not merely seek a solution of $U$ (where $\E U = 0$); it is driven by an emergent force to seek the most stable solution, where the fluctuation $\E[\|U\|^2]$ is minimized.
%Specifically, this force penalizes regions with high curvature where stochastic noise would trigger large gradient responses (high dissipation). 
Thus, the algorithm implicitly regularizes itself, steering the trajectory away from sharp, unstable basins and condensing into robust regimes where the optimization process is thermodynamically efficient. The identity $\phi_{\rm ST} \propto \phi_{\rm DE}$ confirms that this thermodynamic drive for stability is the physical dual of minimizing numerical discretization error. In a metaphorical sense, the learning algorithms seek the ``coldest" minimum. 

\begin{figure}
    \centering
    \includegraphics[width=0.7\linewidth]{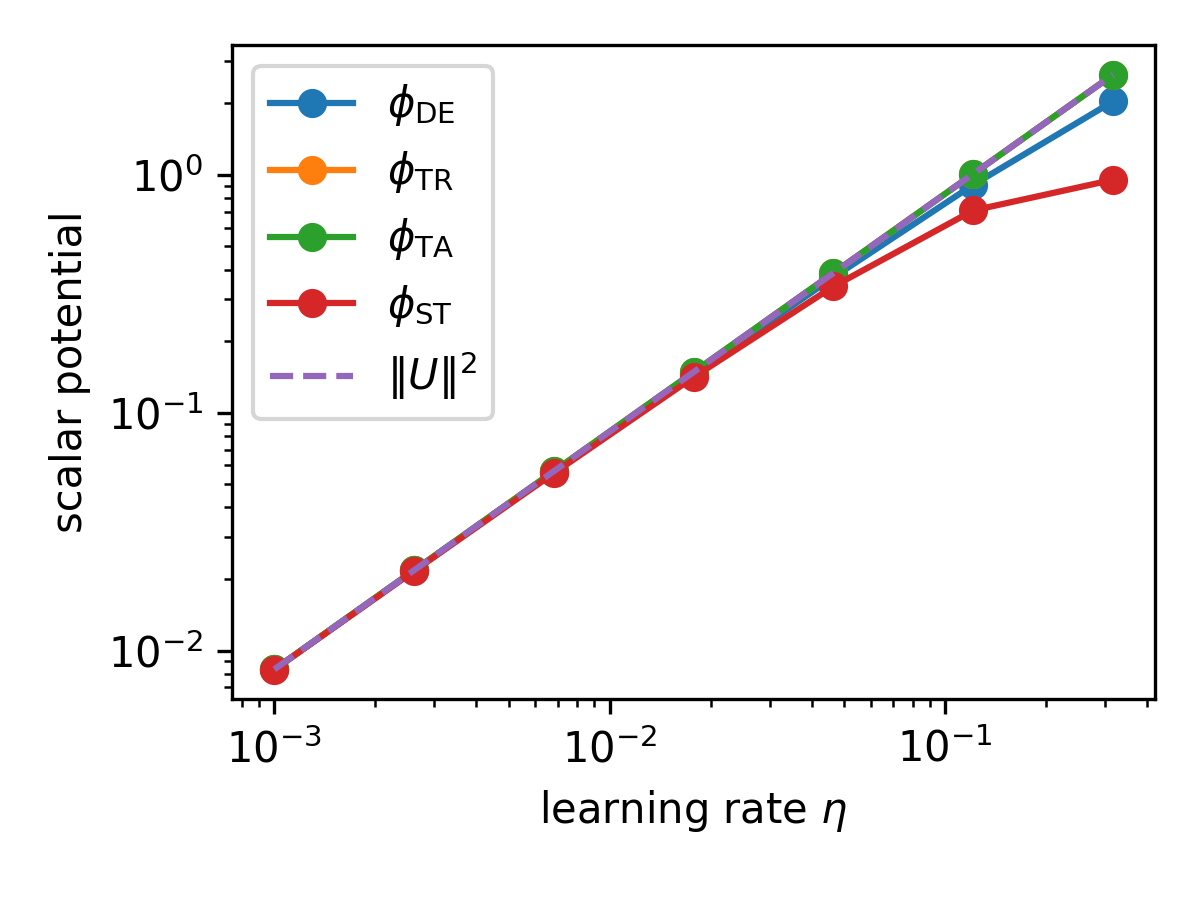}
    \vspace{-1em}
    \caption{All four potentials agree with each other at a small learning rate. As the learning rates increase, they start to deviate, but still have the qualitative agreement that a higher learning rate corresponds to a higher potential, signaling that they are all qualitatively good metrics of irreversibility of the training process.}
    \label{fig:potentials}
\end{figure}

\paragraph{Equivalence of Potentials.} We numerically show that all four potentials are equivalent at a small learning rate. We train parameters with SGD across different learning rates. We measure $\phi_{\rm TA,\ DE,\ TR}$ using numerical integrals over the training paths, and $\phi_{\rm ST}$ directly using the definition. See Supplementary for measurement details. See Figure~\ref{fig:potentials}.

\begin{figure*}[t!]
    \centering
    \includegraphics[width=0.235\linewidth]{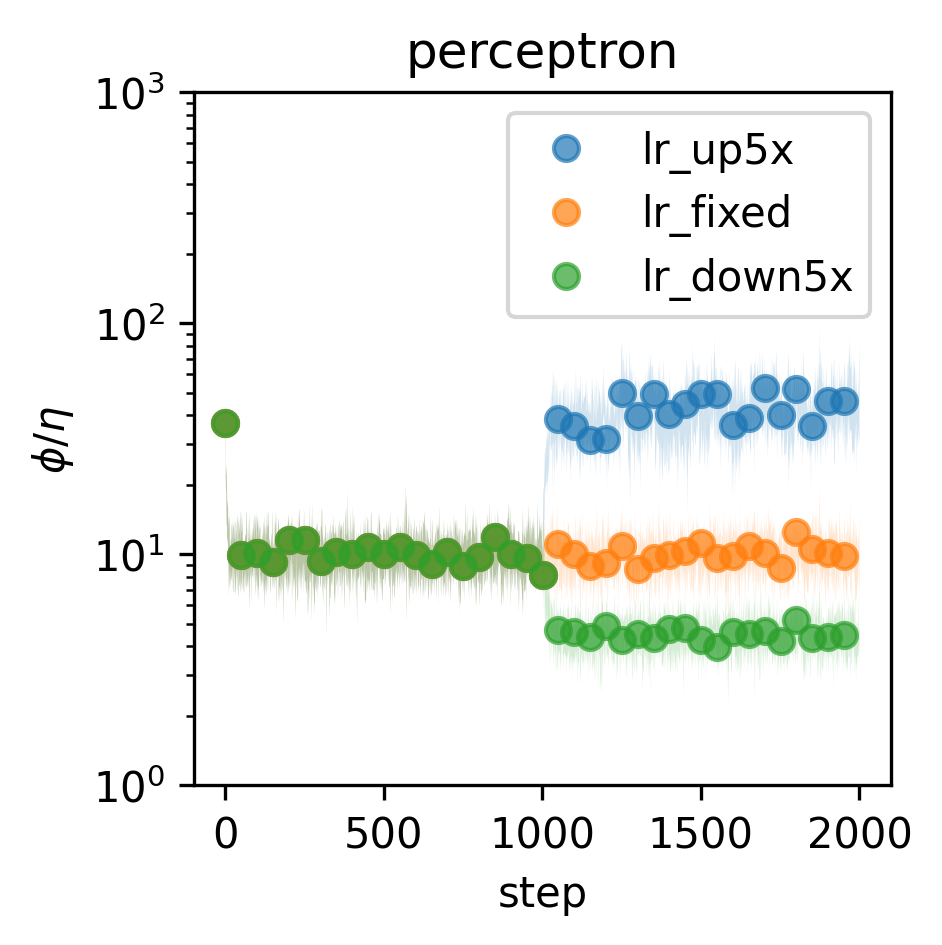}
    \includegraphics[width=0.235\linewidth]{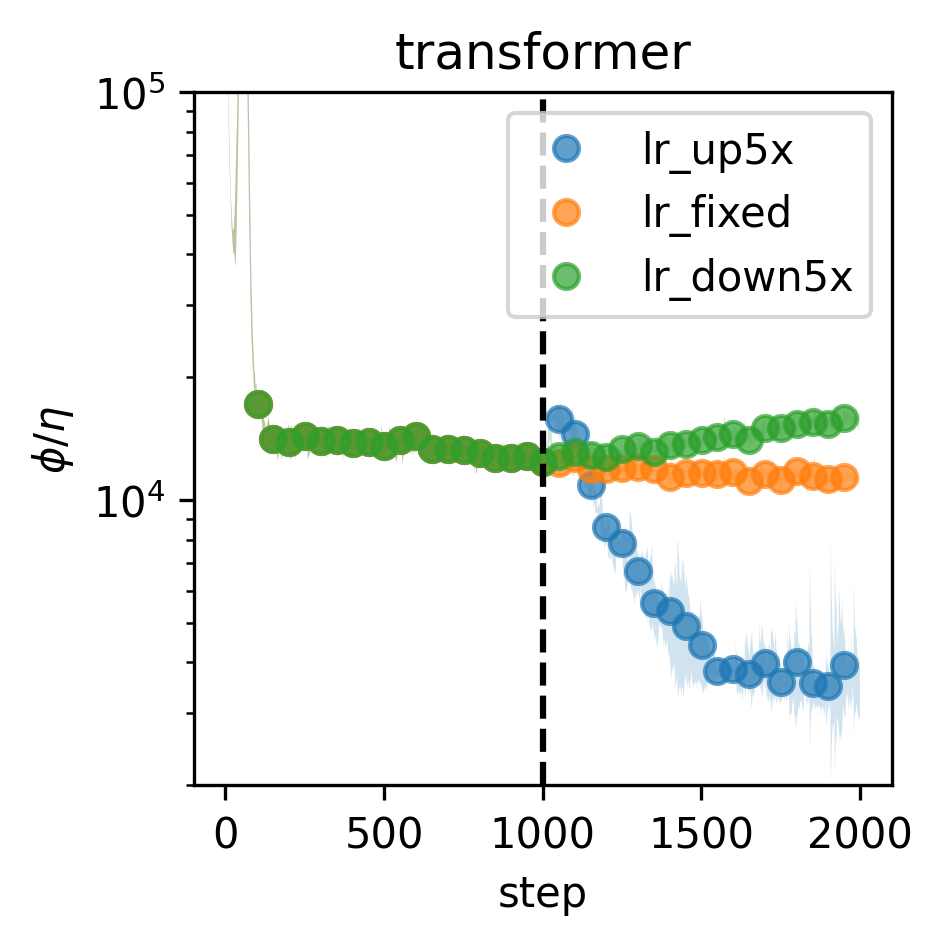}
    \includegraphics[width=0.235\linewidth]{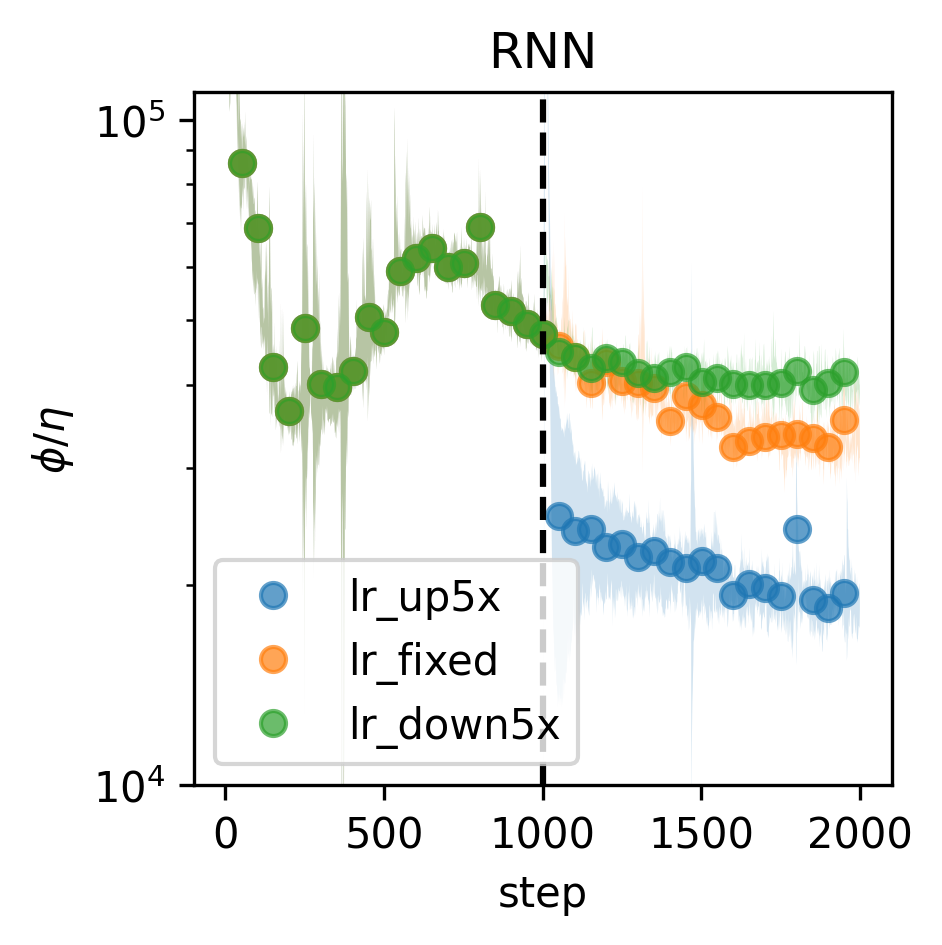}
    \vspace{-1em}
    \caption{Anomalous fluctuation effect in transformer training dynamics. There is a qualitative distinction between the effect of changing the learning rate in a simple linear model (perceptron) and in nonlinear models such as transformers. In linear models, the classical optimization theory predicts an increase in the fluctuation of the gradient due to increased instability \cite{liu2021noise} -- because there is a single local minimum, and so it is impossible for the model to move to a different, low-fluctuation model. In sharp contrast, the nonlinear-model experiments reported here show lower update fluctuation at larger learning rate, in agreement with the prediction of the irreversibility theory.}
    \label{fig: exp}
\end{figure*}

\paragraph{Symmetry Breaking and Preservation.}
From the physical perspective, a key effect of irreversibility is that the effective potential $\phi_{\rm DE}$ alters the inherent symmetries of $U$. When $U$ is GD, this is easy to see \footnote{See \cite{ziyin2025neural} for a prior derivation of this result when specialized to GD.}. Although $L$ may be invariant under a transformation of the coordinates, the Euclidean norm appearing in $\phi_{\rm DE}=\eta\|U\|^2/4$ is not invariant under a general non-orthogonal reparametrization. Therefore, common training algorithms do not have the reparametrization invariance of nature.

This also holds true for a general $U$. Let $K: \mathbb{R}^d \to \mathbb{R}^d$ denote an invertible coordinate transformation. We define the vector field $U$ to be symmetric under $K$ if it transforms covariantly as a gradient field:
\begin{equation}
    J_K(\theta)^T U(K(\theta)) = U(\theta),
\end{equation}
where $J_K(\theta) = \frac{\partial K}{\partial \theta}$ is the Jacobian matrix of $K$. 

For a continuous symmetry parameterized by $\lambda$ and generated by a generic vector field $Q(\theta)$, the transformation expands as $K(\theta, \lambda) = \theta + \lambda Q(\theta) + O(\lambda^2)$. One finds that the effective potential $\phi_{\rm DE}$ explicitly breaks this continuous symmetry unless the quadratic form of the symmetric part of the generator's Jacobian vanishes along $U$, satisfying the condition 
\begin{align}
U(\theta)^T (J_Q(\theta) + J_Q(\theta)^T) U(\theta) = 0.
\end{align}
 The formal proof of this constraint is provided in the supplemental material.

Geometrically, this implies that continuous scaling symmetries---which possess generators with non-vanishing symmetric Jacobians---are generically broken by discretization error, forcing the trajectory to dynamically select a definitive scale. Conversely, continuous rigid rotations (where $J_Q$ is skew-symmetric, i.e., $J_Q + J_Q^T = 0$) are automatically preserved to first order. 

In contrast, discrete linear orthogonal symmetries are preserved by $\phi_{\rm DE}$. If we consider a linear transformation $K(\theta) = O\theta$ defined by an orthogonal matrix $O$ satisfying $O^T O = I$, and $U$ is symmetric under this transformation, the effective potential $\phi_{\rm DE}$ remains strictly invariant. As detailed in the supplemental material, this ensures that these orthogonal discrete symmetries are unaffected by the leading finite-step-size potential. Recent works showed that there is a huge number of discrete symmetries in deep learning \cite{ziyin2025parameter}, and so this implies that (perhaps generalized forms of) mechanisms of spontaneous symmetry breaking can still happen in deep learning, and could be related to the widely observed phenomena of rapid learning. %[todo: references to parameter symmetries]

%[todo: perhaps also an SDE formalism?]

%\paragraph{Numerical Experiments.} %We  demonstrate a direct effect of the theory, which we call the ``anomalous fluctuation," which shows that the fluctuation of the random forces decreases as we increase the learning rate. 
%The details of the experiments are given in the supplemental material.

\paragraph{Anomalous Fluctuation in Transformer Training.} As a simple but important application, We analyze how the magnitude of Adam parameter updates evolves during training for a small decoder-only Transformer on an algorithmic task. 
The model is a $2$-layer causal Transformer (GPT-style) with $d_{\text{model}}=128$, $n_{\text{head}}=4$ attention heads, and a feedforward dimension of $512$. 
Token embeddings and learned positional embeddings are used, followed by a linear output head to logits over a vocabulary of size $V=64$. We train on an autoregressive \emph{sequence reversal} task. 
Each example samples an i.i.d.\ token sequence $x=(x_1,\dots,x_T)$ of length $T=16$ uniformly from $\{4,\dots,V-1\}$. The target sequence is $y=(x_T,\dots,x_1)$. % We concatenate input and target into a single stream
%\[
%[\texttt{BOS}],\, x_1,\dots,x_T,\, [\texttt{SEP}],\, y_1,\dots,y_T,\, [\texttt{EOS}],
%\]
%and train by next-token prediction with cross-entropy loss \emph{masked to only the target segment} (i.e.\ positions after \texttt{SEP}, including \texttt{EOS}). 
To induce label stochasticity, we independently corrupt each target token $y_t$ with probability $p=0.05$, replacing it with a uniformly sampled token from $\{4,\dots,V-1\}$.

We train the model with the same learning rate up to $10^3$ iterations, and either increase or decrease the learning rate by five times and continue training. See Figure~\ref{fig: exp}. We see that there is a qualitative distinction between the effect of changing the learning rate in a simple linear model (perceptron) and in nonlinear models such as transformers. The linear model is observed to obey the normal dissipation fluctuation relations, where a larger learning rate leads to a larger fluctuation \cite{yaida2018fluctuation, liu2021noise}, whereas the deep neural networks exhibit the opposite effect, where a larger learning rate leads to a lower fluctuation. This is because in a linear model, the global minimum is unique and the first term of $F_{\rm eff}$ dominates, and so the increased oscillations reflect the loss of stability in the same basin of attraction. In contrast, for a large nonlinear model, there exist many local minima and degenerate directions of the loss function, and in these flatter directions the emergent entropic term dominates, thereby taking the model to a low-fluctuation state. %[todo: discuss more]

%[todo]

\paragraph{Discussion.} In this work, we have studied four ways to characterize the time irreversibility of a training algorithm of AI models: one comes from the conventional numerical analysis and machine learning literature, and the other three from different branches of theoretical physics. We showed that the scalar potentials $\phi_{\rm DE}$, $\phi_{\rm TR}$, and $\phi_{\rm ST}$ agree to leading order in $\eta$, while $\phi_{\rm TA}$ is controlled by the gradient of this same potential; here $\eta$ is the step size of the algorithm and the arrow of time for the underlying stochastic process induced by the algorithm. This result has the potential to establish a principled study of the irreversibility of training algorithms and thereby allows us to build even closer connections between modern physics and the theory of deep learning.  
Other recent works in machine learning have suggested how the irreversibility of the training process leads to various emergent phenomena of well-trained large neural networks \cite{ziyin2025neural, xu2026does}. On the application side, an important next step would be to identify more AI phenomena due to the irreversibility of training. On the theory side, it would be interesting to treat the irreversibility as an algorithmic-stability metric and establish a direct connection between the irreversibility and the generalization capability of the models \cite{poggio2004general}.

%[todo] %In this work, we have studied

% -------------------- Acknowledgments --------------------
\begin{acknowledgments}
We thank Tomaso Poggio, Yizhou Xu for helpful discussions. ILC acknowledges support in part from the Institute for Artificial Intelligence and Fundamental Interactions (IAIFI) through NSF Grant No. PHY-2019786.
\end{acknowledgments}

% -------------------- Bibliography --------------------
% Option 1: BibTeX
\bibliographystyle{apsrev4-2}
\bibliography{ref}

\appendix
\def\SupplementalMaterialImported{1}
\ifdefined\SupplementalMaterialImported
\else
\documentclass[onecolumn,
 reprint, % two-column PRL style
 amsmath,amssymb,
 aps,
 prl,
 superscriptaddress
]{revtex4-2}

% --- Packages (keep it lean) ---
\usepackage{graphicx}
\usepackage{bm}
\usepackage{hyperref} % 
\usepackage{physics} % optional; remove if you prefer plain LaTeX
\IfFileExists{siunitx.sty}{\usepackage{siunitx}}{} % optional; units
\usepackage{xcolor} % for named colors like orange

\usepackage{amsmath,amsthm,verbatim,amssymb,amsfonts,amscd, graphicx, enumitem}

% --- Additional Packages for Appendix ---
\IfFileExists{tikz.sty}{\usepackage{tikz}}{}
\IfFileExists{centernot.sty}{\usepackage{centernot}}{}
\IfFileExists{footmisc.sty}{\usepackage[bottom]{footmisc}}{}
\IfFileExists{caption.sty}{\usepackage{caption}}{}
\IfFileExists{subcaption.sty}{\usepackage{subcaption}}{}
\IfFileExists{helvet.sty}{\usepackage{helvet}}{}
\usepackage{url}
\IfFileExists{multirow.sty}{\usepackage{multirow}}{}
\IfFileExists{MnSymbol.sty}{\usepackage{MnSymbol}}{}
\IfFileExists{makecell.sty}{\usepackage{makecell}}{}
\IfFileExists{arydshln.sty}{\usepackage{arydshln}}{}
\IfFileExists{bbm.sty}{\usepackage{bbm}}{}
\IfFileExists{textcomp.sty}{\usepackage{textcomp}}{}
\IfFileExists{wrapfig.sty}{\usepackage{wrapfig}}{}
\IfFileExists{algorithm.sty}{\usepackage{algorithm}}{}
\IfFileExists{algorithmic.sty}{\usepackage{algorithmic}}{}
\IfFileExists{csquotes.sty}{\usepackage{csquotes}}{}

\newcommand{\E}{\mathbb{E}}
\newcommand{\V}{{\rm Var}}
\newcommand{\sgn}{{\rm sgn}}

% Custom commands
\newcommand{\R}{\mathbb{R}}
\newcommand{\LR}{\Lambda}
\renewcommand{\grad}{\nabla}
\newcommand{\Hess}{\nabla^2}

\newcommand{\YR}[1]{{\color{orange}[YR:#1]}}

% Theorem environments
\theoremstyle{plain}
\newtheorem{theorem}{Theorem}
\newtheorem{corollary}{Corollary}
\newtheorem{lemma}{Lemma}
\newtheorem*{remark}{Remark}
\newtheorem{proposition}{Proposition}
\newtheorem*{surfacecor}{Corollary 1}
\newtheorem{conjecture}{Conjecture}
\newtheorem{question}{Question} 
\newtheorem{definition}{Definition}
\newtheorem{assumption}{Assumption}

\begin{document}
\fi

\setcounter{secnumdepth}{3}

\section{Notations and setup}
Let $U(\theta)$ be a vector field representing the update direction.
We denote the Jacobian of $U$ as $J_U$, where $[J_U]_{ij} = \partial_j U_i$.
We assume that $U$ is sufficiently smooth (for example, $C^2$ with the boundedness needed for the Taylor remainders below), globally Lipschitz continuous, and has symmetric Jacobian $J_U=J_U^T$. Since the parameter space is taken to be $\mathbb{R}^d$, this symmetric-Jacobian assumption makes $U$ conservative up to the usual regularity assumptions.
The discrete-time update rule with step size (learning rate) $\eta > 0$ is defined as:
\begin{equation}\label{eq:discrete_dynamics}
 \theta_{k+1} = \theta_k - \eta U(\theta_k).
\end{equation}
We denote the parameter obtained after $t$ iterations starting from $\theta_0$ as $\theta_t^U$. The effective potential $\phi_{\rm DE}$ derived in the following sections is defined under this conservative/symmetric-Jacobian assumption; for non-conservative fields, the leading discretization correction is still a vector field but cannot generally be expressed as the gradient of a scalar potential.

To analyze the discrete dynamics, we compare it to a reference continuous-time flow defined by the ordinary differential equation (ODE):
\begin{equation}\label{eq:continuous_dynamics}
\dot{\theta}(\tau) = -\eta U(\theta(\tau)).
\end{equation}
Here, the time variable $\tau$ in the continuous flow corresponds directly to the iteration count in the discrete dynamics. We denote the solution at time $\tau$ starting from $\theta_0$ as $\bar{\theta}_\tau^U$.

\section{Discretization Error}

In this section, we analyze the deviations of the discrete dynamics from the continuous flow. We derive the effective dynamics that describe the discrete update to higher order and establish the connection between the discretization error and the breakdown of time-reversibility. The analysis in this section is conventionally known as backward error analysis in numerical analysis.

Backward Error Analysis (BEA) seeks a modified vector field, $U_{\text{eff}} = U + \nabla \phi_{\rm DE}$, such that the flow of $U_{\text{eff}}$ approximates the discrete dynamics to a higher order.

\begin{theorem}[Discretization Error Potential $\phi_{\rm DE}$]\label{theo: effective dynamics}
Assume that the vector field $U(\theta)$ is sufficiently smooth and has symmetric Jacobian $J_U$.
There exists an effective vector field $U_{\text{eff}}(\theta) = U(\theta) + \nabla \phi_{\rm DE}(\theta)$ such that over a fixed iteration horizon, to leading order, the discrete trajectory $\theta_t^U$ satisfies
\begin{equation}
 \theta_t^U = \bar{\theta}_t^{U+\nabla \phi_{\rm DE}} + O(\eta^2),
\end{equation}
where the flow is defined by $\dot{\theta} = -\eta U_{\text{eff}}(\theta)$.
The leading order Discretization Error Potential $\phi_{\rm DE}$ is given by:
\begin{equation}
 \phi_{\rm DE}(\theta) = \frac{\eta}{4} ||U(\theta)||^2.
\end{equation}
\end{theorem}

\begin{proof}
We perform BEA over a single step. We seek a field $V(\theta) = U(\theta) + f(\theta)$, where $f=O(\eta)$, such that the continuous flow $\bar{\theta}_1^V$ matches the discrete update $\theta_1^U$ up to $O(\eta^3)$.
Discrete update: $\theta_1^U = \theta_0 - \eta U_0$.
Continuous flow expansion for $\dot{\theta} = -\eta V$:
\begin{align}
 \bar{\theta}_1^V &= \theta_0 + \dot{\theta}(0) + \frac{1}{2}\ddot{\theta}(0) + O(\eta^3) \\
 &= \theta_0 - \eta V_0 + \frac{1}{2} (-\eta J_V)(-\eta V_0) + O(\eta^3) \\
 &= \theta_0 - \eta V_0 + \frac{\eta^2}{2} J_V V_0 + O(\eta^3).
\end{align}
Substituting $V = U + f$ and matching $\bar{\theta}_1^V = \theta_1^U$:
\begin{equation}
 -\eta (U_0 + f_0) + \frac{\eta^2}{2} J_U U_0 = -\eta U_0 \implies f_0 = \frac{\eta}{2} J_U U_0.
\end{equation}
Since $J_U$ is assumed symmetric, we have $J_U U = \nabla (\frac{1}{2} ||U||^2)$.
We can thus rewrite the force $f$ as a gradient:
\begin{equation}
 f = \frac{\eta}{2} J_U U = \frac{\eta}{2} \nabla \left( \frac{1}{2} ||U||^2 \right) = \nabla \left( \frac{\eta}{4} ||U||^2 \right).
\end{equation}
Thus, $\phi_{\rm DE} = \frac{\eta}{4} ||U||^2$.
\end{proof}
This discussion is related to Lemma~4 and 5 in Ref.~\cite{li2018stochasticmodifiedequationsdynamics}, where a different, more stochastic approach is used in the derivation.

\subsection{Time Reversal Asymmetry}

\begin{definition}[Time Reversal Asymmetry $\phi_{\rm TA}$]
We quantify the microscopic time reversal asymmetry of the discrete update rule using the difference between the initial parameter state and the state recovered by sequentially applying the forward and backward dynamics:
\begin{equation}
    2\eta \nabla\phi_{\rm TA}(\theta_t) := \theta_t - \tilde{\theta}_t,
\end{equation}
where the recovered state is $\tilde{\theta}_t = \theta_t - \eta U(\theta_t) + \eta U(\theta_{t+1})$.
\end{definition}
\begin{theorem}[Relationship between $\phi_{\rm TA}$ and $\phi_{\rm DE}$]\label{theo:drift}
Assuming the vector field $U$ possesses a symmetric Jacobian $J_U$, the time reversal asymmetry is related to the gradient of $\phi_{\rm DE}$ by
\begin{align}
    \phi_{\rm TA}(\theta_t) = \phi_{\rm DE}(\theta_t) + \mathcal{O}(\eta^2).
\end{align}
\end{theorem}
\begin{proof}
Using $\theta_{t+1}=\theta_t-\eta U(\theta_t)$, Taylor expansion gives
\begin{align}
    U(\theta_{t+1}) &= U(\theta_t-\eta U(\theta_t)) \nonumber\\
    &=U(\theta_t)-\eta J_U(\theta_t)U(\theta_t)+\mathcal{O}(\eta^2).
\end{align}
Therefore,
\begin{align}
    \tilde\theta_t &= \theta_t-\eta U(\theta_t)+\eta U(\theta_{t+1}) \nonumber\\
    &=\theta_t-\eta^2J_U(\theta_t)U(\theta_t)+\mathcal{O}(\eta^3),
\end{align}
and hence
\begin{equation}
    2 \eta  \nabla \phi_{\rm TA}(\theta_t)=\theta_t-\tilde\theta_t= \eta^2 J_U(\theta_t)U(\theta_t)+\mathcal{O}(\eta^3).
\end{equation}
This implies that 
\begin{equation}
    \phi_{\rm TA}(\theta_t) = \frac{\eta}{4}\|U\|^2 + \mathcal{O}(\eta^2)
\end{equation}
Finally, comparing with Theorem~\ref{theo: effective dynamics} gives the desired result.
\end{proof}

\section{Time Renormalization}
We now derive the potential $\phi_{\rm TR}$ through a renormalization group (RG) analysis of the time discretization. The central idea is to view the discrete update rule $U$ not as a fundamental atomic operation, but as the coarse-grained effective dynamics of a finer underlying process. We ask: what correction must be added to a fine-grained process so that, after coarse-graining, it reproduces the original dynamics?

Let $\Phi_\eta^U(\theta) = \theta - \eta U(\theta)$ denote the discrete update map. We define a renormalization step as doubling the temporal resolution (halving the step size) while adjusting the vector field to preserve the macroscopic flow.

\begin{theorem}[Time Renormalization Potential $\phi_{\rm TR}$]\label{theo:renormalization}
Let $U_0 = U$ be the baseline update rule with step size $\eta$. Consider a sequence of refined vector fields $\{U_k\}_{k=0}^\infty$ where $U_k$ operates at step size $\eta_k = \eta / 2^k$. If $U_{k+1}$ is chosen such that applying it twice at step $\eta_{k+1}$ approximates one step of $U_k$ at $\eta_k$ with local error $O(\eta_k^3)$, i.e., 
\begin{align}
\Phi_{\eta_{k+1}}^{U_{k+1}} \circ \Phi_{\eta_{k+1}}^{U_{k+1}} \simeq \Phi_{\eta_k}^{U_k},\quad 
\Phi_\eta^U(\theta) := \theta - \eta U(\theta),
\end{align}
 then the cumulative correction required to reach the continuous limit $U_\infty$ corresponds to the potential $\phi_{\rm TR} = \frac{\eta}{4} ||U||^2$.
Specifically, assuming the symmetric-Jacobian condition persists along the refinement, $U_\infty = U + \nabla \phi_{\rm TR}$ to leading order.
\end{theorem}

\begin{proof}
We first derive the correction for a single doubling step. Throughout the derivation we keep terms through order $\delta^2$ and absorb higher-order terms into the remainder.
Let $U_k$ be the field at step $\delta = \eta_k$. We seek $U_{k+1}$ at step $\delta/2$ such that:
\begin{equation}
 \Phi_{\delta/2}^{U_{k+1}} \left( \Phi_{\delta/2}^{U_{k+1}}(\theta) \right) = \Phi_{\delta}^{U_k}(\theta) + O(\delta^3).
\end{equation}
Let $U_{k+1} = U_k + \Delta U$. Expanding the composition of two fine steps:
\begin{align}
 \theta_{mid} &= \theta - \frac{\delta}{2} U_{k+1}(\theta), \\
 \theta_{next} &= \theta_{mid} - \frac{\delta}{2} U_{k+1}(\theta_{mid}) \nonumber \\
 &= \theta - \frac{\delta}{2} U_{k+1} - \frac{\delta}{2} \left( U_{k+1} - \frac{\delta}{2} J_{U_{k+1}} U_{k+1} \right) +O(\delta^3) \nonumber \\
 &= \theta - \delta U_{k+1} + \frac{\delta^2}{4} J_{U_{k+1}} U_{k+1}+O(\delta^3) .
\end{align}
Since $\Delta U=O(\delta)$, substituting $U_{k+1}= U_k + \Delta U$ gives:
\begin{equation}
 \theta_{next}=\theta - \delta (U_k + \Delta U) + \frac{\delta^2}{4} J_{U_k} U_{k}+O(\delta^3).
\end{equation}
We equate this to the coarse update $\theta - \delta U_k$:
\begin{equation}
 -\delta \Delta U + \frac{\delta^2}{4} J_{U_k} U_k = 0 \implies \Delta U = \frac{\delta}{4} J_{U_k} U_k.
\end{equation}
Since $J_{U_k}$ is symmetric by assumption, we have $J_{U_k} U_k = \nabla (\frac{1}{2} ||U_k||^2)$.
Thus, the correction vector is a gradient of a potential:
\begin{equation}
 \Delta U = \nabla \left( \frac{\delta}{8} ||U_k||^2 \right).
\end{equation}
This confirms the coefficient $1/8$ for a single renormalization step.
Now we sum the corrections over all scales $k=0$ to $\infty$. The step size at level $k$ is $\delta = \eta/2^k$. The update to the potential at step $k$ is:
\begin{equation}
 \Delta \phi_k = \frac{\eta/2^k}{8} ||U||^2.
\end{equation}
(We approximate $U_k \approx U$ in the magnitude term as higher order corrections to $U$ contribute negligibly to the potential coefficient).
The cumulative renormalization potential after $m$ refinements is the partial geometric sum
\begin{align}
 \phi_{\rm TR}^{(m)} =& \sum_{k=0}^{m-1} \Delta \phi_k
 = \sum_{k=0}^{m-1} \frac{\eta}{8 \cdot 2^k} ||U||^2\nonumber\\
 = &\frac{\eta}{4}(1-2^{-m})||U||^2+O(\eta^2). 
\end{align}
Taking $m\to\infty$ gives
\begin{equation}
 \phi_{\rm TR} = \frac{\eta}{4} ||U||^2,
\end{equation}
which matches the discretization error potential $\phi_{\rm DE}$ to leading order.
\end{proof}

\section{Entropy Production}

%\subsection{SGD is Fully Absolutely Irreversible}

%[todo]

%\subsection{Regularized EP}
We analyze the thermodynamic entropy production in the presence of stochastic noise. We consider the discrete-time Langevin dynamics of model parameters $\theta \in \mathbb{R}^d$:
\begin{equation}
 \theta_{k+1} = \theta_k - \eta U(\theta_k) + \sigma \varepsilon_k, \quad \varepsilon_k \sim \mathcal{N}(0, I),
\end{equation}
where $\sigma$ determines the noise scale. To map this to a thermodynamic process, we operate in the scaling regime where $\sigma^2 \propto \eta$.
The total entropy production along a trajectory is the sum of the entropy change of the bath (often called ``heat") and the information entropy change of the system:
\begin{equation}
 \Delta S_{\text{tot}} = \Delta S_{\text{bath}} + \Delta S_{\text{sys}}.
\end{equation}
To evaluate the entropy production of a discrete update, we assume the initial state at step $k$ follows a Gaussian distribution $\rho_k(\theta) = \mathcal{N}(\theta | \mu, \tau I)$ centered at $\mu$ with variance $\tau$. We analyze the thermodynamic forces by taking the limits $\sigma \to 0$ and $\tau \to 0$ sequentially.

\begin{proposition}[Vanishing System Entropy Contribution]
Consider the noisy update $\theta_{k+1} = \theta_k - \eta U(\theta_k) + \sigma \varepsilon$ initialized from $\rho_k = \mathcal{N}(\mu, \tau I)$. In the scaled limit $\sigma \to 0$ (keeping $\tau$ fixed), the system entropy production $\Delta S_{\text{sys}}$ vanishes when scaled by the noise temperature $\sigma^2$.
\end{proposition}

\begin{proof}
The Shannon entropy of the initial Gaussian state is $H(\rho_k) = \frac{d}{2}(1+\ln(2\pi)) + \frac{d}{2}\ln \tau$. The updated state $\theta_{k+1}$ has covariance $\Sigma_{k+1} = (\tau + \sigma^2)I - \eta \tau (J_U + J_U^T) + O(\tau \eta^2)$.
The change in system entropy is given by:
\begin{align}
\Delta S_{\text{sys}} &= \frac{1}{2} \ln \det \Sigma_{k+1} - \frac{d}{2} \ln \tau \nonumber \\
&= \frac{d}{2} \ln \left( 1 + \frac{\sigma^2}{\tau} \right) - \frac{\eta \tau}{\tau+\sigma^2} \nabla \cdot U + O(\eta^2).
\end{align}
Multiplying by $\sigma^2$ and taking the limit $\sigma \to 0$ with $\tau > 0$ fixed:
\begin{equation}
\lim_{\sigma^2 \to 0} \sigma^2 \Delta S_{\text{sys}} = \lim_{\sigma^2 \to 0} \left[ \sigma^2 \frac{d}{2} \ln \left( 1 + \frac{\sigma^2}{\tau} \right) - \sigma^2 \frac{\eta \tau}{\tau+\sigma^2} \nabla \cdot U \right]. \nonumber
\end{equation}
Since $\lim_{x\to 0} x \ln(1+x/C) = 0$, both terms vanish. Thus, the system entropy does not contribute to the scaled entropy production in this limit.
\end{proof}

\begin{theorem}[Stochastic Entropy Potential $\phi_{\rm ST}$]\label{theo:entropy_production_new}
In the limit $\tau \to 0$ of the scaled entropy production, the effective thermodynamic driving force is governed by the potential $\phi_{\rm ST} = \frac{\eta}{4} ||U||^2$. Specifically:
\begin{equation}
 \lim_{\tau \to 0} \left( \lim_{\sigma^2 \to 0} \sigma^2 \mathbb{E}[\Delta S_{\rm{tot}}] \right) = 8\eta \phi_{\rm ST}(\mu)+O(\eta^3).
\end{equation}
\end{theorem}

\begin{proof}
Since the system entropy contribution vanishes in the scaled limit, we focus on the bath entropy production $\Delta S_{\text{bath}}$, defined by the log-ratio of forward to backward transition probabilities:
\begin{align}
 &\Delta S_{\text{bath}} = \ln \frac{p(\theta_{k+1}|\theta_k)}{p(\theta_k|\theta_{k+1})} \nonumber\\
 =& \frac{1}{2\sigma^2} \left[ ||\theta_k - \theta_{k+1} + \eta U(\theta_{k+1})||^2 - ||\theta_{k+1} - \theta_k + \eta U(\theta_k)||^2 \right].\nonumber
\end{align}
Let $\Delta = \theta_{k+1} - \theta_k = -\eta U(\theta_k) + \sigma \varepsilon$. Expanding the squares and Taylor expanding $U(\theta_{k+1})$ around $\theta_k$, we obtain (as derived in standard stochastic thermodynamics):
\begin{align}
 &\Delta S_{\text{bath}} \nonumber\\
 =& \frac{1}{2\sigma^2} \Big( -4\eta \Delta^T U - 2\eta \Delta^T J_U \Delta + 2\eta^2 U^T J_U \Delta \Big) 
 + R_{\eta,\sigma}, \nonumber
\end{align}
where $R_{\eta,\sigma}$ is such that $\sigma^2 \E[R_{\eta,\sigma}] = O(\eta^3)$.
Taking the expectation over the noise $\varepsilon$ (where $\E[\Delta] = -\eta U$ and $\E[\Delta \Delta^T] = \eta^2 UU^T + \sigma^2 I$):
\begin{align}
 \mathbb{E}[\Delta S_{\text{bath}}] &= \frac{1}{2\sigma^2} \left( 4\eta^2 ||U||^2 - 2\eta \sigma^2 \nabla \cdot U \right) + O(\eta^3/\sigma^2) \nonumber \\
 &= \frac{2\eta^2}{\sigma^2} ||U(\theta_k)||^2 - \eta \nabla \cdot U(\theta_k) + O(\eta^3/\sigma^2).\nonumber
\end{align}
Now we apply the limits. First, multiply by $\sigma^2$ and let $\sigma \to 0$:
\begin{equation}
 \lim_{\sigma^2 \to 0} \sigma^2 \mathbb{E}[\Delta S_{\text{bath}}] = 2\eta^2 ||U(\theta_k)||^2 + O(\eta^3).
\end{equation}
Finally, we take the limit $\tau \to 0$, where the distribution of $\theta_k$ concentrates on $\mu$:
\begin{align}
 &\lim_{\tau \to 0} \E_{\theta_k \sim \mathcal{N}(\mu, \tau I)} [2\eta^2 ||U(\theta_k)||^2 + O(\eta^3)] \nonumber\\
 =& 2\eta^2 ||U(\mu)||^2 + O(\eta^3).
\end{align}
Identifying the potential $\phi_{\rm ST}(\mu) = \frac{\eta}{4} ||U(\mu)||^2$, we have $2\eta^2 ||U(\mu)||^2 = 8\eta \phi_{\rm ST}(\mu)$, up to higher-order terms in $\eta$.
\end{proof}

\section{The Equivalence Theorem}

We summarize the central result connecting the different sources of irreversibility.
\begin{theorem}[Equivalence of Corrections]
Under the symmetric-Jacobian assumptions,
\begin{align}
 \phi_X(\theta) = \frac{\eta}{4}\|U(\theta)\|^2 + O(\eta^2),
\end{align}
for all $ X\in\{{\rm DE},{\rm TA},{\rm TR},{\rm ST}\}$.
%Furthermore, the Time Reversal Asymmetry $\phi_{\rm TA}$ (from the breakdown of forward-backward reversibility) is structurally linked to the gradient of this unified potential by:
%\begin{equation}
%\phi_{\rm TA} = 4\eta^2 \|\nabla \phi_{\rm DE}\|^2 + \mathcal{O}(\eta^5).
%\end{equation}
%The Asymmetric Heat Potential is $\tilde{\phi} = 2 \phi_{\rm DE}$. Thus, within the stated leading-order and symmetric-Jacobian assumptions, the scalar-potential formulations of deterministic discretization, stochastic modified-equation approximation, time-scale renormalization, and regularized stochastic thermodynamics coincide, while trajectory irreversibility is controlled by the gradient of this potential.
\end{theorem}

\section{Symmetry Breaking and Preservation of the Effective Dynamics}

We analyze the symmetries of the effective dynamics derived in the supplemental material. We assume the dynamics are governed by a vector field $U: \mathbb{R}^d \to \mathbb{R}^d$ with a symmetric Jacobian $J_U(\theta) = J_U(\theta)^T$. Due to the equivalence theorem, it is sufficient to just consider $\phi_{\rm DE}$.

The discrete update rule is $\theta_{k+1} = \theta_k - \eta U(\theta_k)$. As shown in Theorem \ref{theo: effective dynamics} of the supplemental material, the effective dynamics are driven by $U_{\text{eff}} = U + \nabla \phi_{\rm DE}$, where the effective potential (discretization error) is:
\begin{equation} \label{eq:phi_de}
 \phi_{\rm DE}(\theta) = \frac{\eta}{4} ||U(\theta)||^2 = \frac{\eta}{4} U(\theta)^T U(\theta).
\end{equation}

\subsection{Symmetry for a Vector Field}

In the context of optimization, a symmetry usually refers to the invariance of the underlying loss function. Since we are working directly with the vector field $U$ (which corresponds to the gradient of the loss when the update is $\theta_{k+1}=\theta_k-\eta U(\theta_k)$), we must define what ``symmetry'' means for $U$ directly.

\begin{definition}[Symmetry Condition for $U$]
Let $K: \mathbb{R}^d \to \mathbb{R}^d$ be an invertible transformation. We say the vector field $U$ is symmetric under $K$ if it transforms covariantly as a gradient field:
\begin{equation} \label{eq:symmetry_def}
 J_K(\theta)^T U(K(\theta)) = U(\theta),
\end{equation}
where $J_K(\theta) = \frac{\partial K}{\partial \theta}$ is the Jacobian matrix of the transformation.
\end{definition}

We justify the definition by looking at a special case that is commonly considered. Suppose $U(\theta) = \nabla L(\theta)$, then $L(K(\theta)) = L(\theta)$ implies $\nabla_\theta [L(K(\theta))] = \nabla_\theta L(\theta)$, which by the chain rule gives $J_K(\theta)^T \nabla L(K(\theta)) = \nabla L(\theta)$. Equation \eqref{eq:symmetry_def} captures this constraint purely in terms of $U$.

\subsection{Continuous Symmetry Breaking}

We now derive the condition under which a continuous symmetry $K(\theta, \lambda)$ is preserved or broken by the effective potential $\phi_{\rm DE}$.
\begin{theorem}[Continuous Symmetry Breaking]
Let $K(\theta, \lambda) = \theta + \lambda Q(\theta) + O(\lambda^2)$ be a continuous symmetry generated by $Q(\theta)$.
If $U(\theta)^T(J_Q(\theta)+J_Q(\theta)^T)U(\theta)\neq 0$, then the entropic potential $\phi_{\rm DE}$ breaks the symmetry to first order in $\lambda$.
\end{theorem}

\begin{proof}
First, we determine how the symmetry of the original dynamics constrains $U$.
The Jacobian of the transformation is $J_K(\theta) = I + \lambda J_Q(\theta) + O(\lambda^2)$, where $J_Q(\theta)$ is the Jacobian of the generator $Q$.
Substituting this expansion into the symmetry definition \eqref{eq:symmetry_def}:
\begin{equation}
 (I + \lambda J_Q(\theta))^T U(\theta + \lambda Q(\theta)) = U(\theta) + O(\lambda^2).
\end{equation}
We expand $U(\theta + \lambda Q)$ using the Taylor series (utilizing the Jacobian $J_U$):
\begin{equation}
 U(\theta + \lambda Q) = U(\theta) + \lambda J_U(\theta) Q(\theta) + O(\lambda^2).
\end{equation}
Substituting this back:
\begin{align}
 (I + \lambda J_Q^T) (U + \lambda J_U Q) &= U \\
 U + \lambda J_U Q + \lambda J_Q^T U + O(\lambda^2) &= U.
\end{align}
Matching terms at order $\lambda$ gives the fundamental constraint on $U$ for it to possess the continuous symmetry generated by $Q$:
\begin{equation} \label{eq:sym_constraint}
 J_U(\theta) Q(\theta) = -J_Q(\theta)^T U(\theta).
\end{equation}

Now we check if the effective potential $\phi_{\rm DE}$ is invariant under this transformation. We calculate the difference $\Delta \phi_{\rm DE} = \phi_{\rm DE}(K(\theta)) - \phi_{\rm DE}(\theta)$.
\begin{align}
 \phi_{\rm DE}(K) &= \frac{\eta}{4} ||U(K)||^2  \nonumber\\
 &= \frac{\eta}{4} (U + \lambda J_U Q)^T (U + \lambda J_U Q) + O(\lambda^2) \nonumber \\
 &= \frac{\eta}{4} \left( U^T U + \lambda U^T J_U Q + \lambda (J_U Q)^T U \right) + O(\lambda^2) \nonumber \\
 &= \phi_{\rm DE}(\theta) + \frac{\lambda \eta}{2} U^T J_U Q + O(\lambda^2).
\end{align}
(Note: $U^T J_U Q$ is a scalar, so it equals its transpose).

We substitute the symmetry constraint derived in the first step, $J_U Q = -J_Q^T U$, into the change in potential:
\begin{align}
 \Delta \phi_{\rm DE} &= \frac{\lambda \eta}{2} U^T (-J_Q^T U) + O(\lambda^2) \\
 &= -\frac{\lambda \eta}{2} U(\theta)^T J_Q(\theta)^T U(\theta) + O(\lambda^2).
\end{align}
The term $U^T J_Q^T U$ is a quadratic form. For any square matrix $A$, $x^T A x = x^T (\frac{A + A^T}{2}) x$. Thus:
\begin{equation}
 \Delta \phi_{\rm DE} = -\frac{\lambda \eta}{2} U(\theta)^T \left( \frac{J_Q(\theta) + J_Q(\theta)^T}{2} \right) U(\theta) + O(\lambda^2).\nonumber
\end{equation}
For the symmetry to be preserved to first order (i.e., $\Delta \phi_{\rm DE} = O(\lambda^2)$), we require the leading term to vanish:
\begin{equation}
 U(\theta)^T \left( J_Q(\theta) + J_Q(\theta)^T \right) U(\theta) = 0.
\end{equation}
This implies that the effective potential breaks the symmetry unless the quadratic form of the symmetric part of the generator's Jacobian vanishes along $U$. A skew-symmetric generator, such as a rigid rotation with $J_Q+J_Q^T=0$, is a sufficient condition for preservation, but it is not necessary because the quadratic form can also vanish accidentally on a given trajectory.
\end{proof}

\subsection{Discrete Symmetry Preservation}

We derive the preservation of discrete orthogonal symmetries.

\begin{theorem}[Discrete Symmetry Preservation]
Let the transformation be $K(\theta) = O\theta$, where $O$ is an orthogonal matrix ($O^T O = I$). If $U$ is symmetric under $K$, then the effective potential $\phi_{\rm DE}$ is invariant under $K$.
\end{theorem}

\begin{proof}
The Jacobian is $J_K(\theta) = O$.

Using definition \eqref{eq:symmetry_def}, the condition for $U$ to be symmetric under $O$ is:
\begin{align}
 O^T U(O\theta) &= U(\theta) \\
 \implies U(O\theta) &= O U(\theta). \label{eq:discrete_constraint}
\end{align}
This means the vector field rotates with the domain.

We evaluate the effective potential at the transformed point $\theta' = O\theta$:
\begin{align}
 \phi_{\rm DE}(O\theta) &= \frac{\eta}{4} ||U(O\theta)||^2  = \frac{\eta}{4} ||O U(\theta)||^2 \\
 &= \frac{\eta}{4} (O U(\theta))^T (O U(\theta)) = \frac{\eta}{4} U(\theta)^T O^T O U(\theta). \nonumber
\end{align}
where we have substituted  Eq.~\eqref{eq:discrete_constraint} in the first line.
Since $O$ is orthogonal, $O^T O = I$:
\begin{equation}
 \phi_{\rm DE}(O\theta) = \frac{\eta}{4} U(\theta)^T U(\theta) = \phi_{\rm DE}(\theta).
\end{equation}
Thus, the effective potential is invariant under discrete orthogonal symmetries; since the gradient of an invariant scalar under an orthogonal map transforms covariantly, the leading effective correction also preserves the symmetry.
\end{proof}

%\subsection{Examples: SGD and Adam}

%[todo]

\section{Experimental Measurement of the Four Irreversibility Potentials}
\label{app:operational_measurements}

In this appendix, we describe how the four irreversibility quantities
\[
\phi_{\rm DE},\qquad
\phi_{\rm TR},\qquad
\phi_{\rm TA},\qquad
\phi_{\rm ST}
\]
are measured operationally in experiments. We consider a discrete-time update rule
\begin{equation}
    \theta_{k+1}
    =
    \theta_k-\eta U(\theta_k),
    \label{eq:app_discrete_update}
\end{equation}
where \(\theta\in\mathbb{R}^d\), \(U(\theta)\) is the update vector field, and \(\eta>0\) is the step size. The goal is to measure each irreversibility quantity directly from finite-step dynamics, rather than merely evaluating the leading-order theoretical expression.

Throughout, when a measured force field \(F(\theta)\) is obtained, we reconstruct a scalar potential by line integration along the straight path from a reference point \(\theta_{\rm ref}\) to \(\theta\):
\begin{equation}
    \widehat{\phi}(\theta)-\widehat{\phi}(\theta_{\rm ref})
    =
    \int_0^1
    F\!\left(\theta_{\rm ref}
    +s(\theta-\theta_{\rm ref})\right)
    \cdot
    (\theta-\theta_{\rm ref})\,ds.
    \label{eq:app_line_integral}
\end{equation}
In the experiments below we take \(\theta_{\rm ref}=0\). Numerically, the integral is approximated by a trapezoidal rule over a finite number of interpolation points. This construction gives an operational scalar potential from the measured local irreversible force.

\subsection{Discretization-error potential \texorpdfstring{\(\phi_{\rm DE}\)}{phiDE}}

The discretization-error potential measures the mismatch between the discrete update and the continuous-time flow generated by the same vector field. Let \(\Theta_{\rm disc}(\theta;\eta)\) denote one discrete Euler step:
\begin{equation}
    \Theta_{\rm disc}(\theta;\eta)
    =
    \theta-\eta U(\theta).
\end{equation}
Let \(\Theta_{\rm cont}(\theta;\eta)\) denote the time-one solution of the continuous-time ODE
\begin{equation}
    \frac{d\vartheta}{d\tau}
    =
    -\eta U(\vartheta),
    \qquad
    \vartheta(0)=\theta,
    \qquad
    \Theta_{\rm cont}(\theta;\eta)=\vartheta(1).
    \label{eq:app_cont_flow}
\end{equation}
The local discretization mismatch is
\begin{equation}
    d_{\rm DE}(\theta;\eta)
    =
    \Theta_{\rm cont}(\theta;\eta)
    -
    \Theta_{\rm disc}(\theta;\eta).
\end{equation}
For small \(\eta\), Taylor expansion gives
\begin{equation}
    d_{\rm DE}(\theta;\eta)
    =
    \frac{\eta^2}{2}J_U(\theta)U(\theta)
    +O(\eta^3),
\end{equation}
where \(J_U\) is the Jacobian of \(U\). Since the leading-order discretization-error force is
\begin{equation}
    \nabla \phi_{\rm DE}(\theta)
    =
    \frac{\eta}{2}J_U(\theta)U(\theta)
    +O(\eta^2),
\end{equation}
we define the operational force estimator
\begin{equation}
    F_{\rm DE}(\theta;\eta)
    =
    \frac{
    \Theta_{\rm cont}(\theta;\eta)
    -
    \Theta_{\rm disc}(\theta;\eta)
    }{\eta}.
    \label{eq:app_F_DE}
\end{equation}
The scalar potential \(\widehat{\phi}_{\rm DE}\) is then obtained by line integration:
\begin{equation}
    \widehat{\phi}_{\rm DE}(\theta)
    =
    \int_0^1
    F_{\rm DE}(s\theta;\eta)\cdot \theta\,ds.
    \label{eq:app_phi_DE_measured}
\end{equation}

\subsection{Time-renormalization potential \texorpdfstring{\(\phi_{\rm TR}\)}{phiTR}}

The time-renormalization potential is measured by comparing one coarse step of size \(\eta\) to two fine steps of size \(\eta/2\). Define
\begin{equation}
    \Theta_{\rm coarse}(\theta;\eta)
    =
    \theta-\eta U(\theta),
\end{equation}
and
\begin{equation}
    \Theta_{\rm fine}(\theta;\eta)
    =
    \Theta_{\eta/2}
    \circ
    \Theta_{\eta/2}(\theta),
    \qquad
    \Theta_{\eta/2}(\theta)
    =
    \theta-\frac{\eta}{2}U(\theta).
\end{equation}
The coarse-fine mismatch is
\begin{equation}
    d_{\rm TR}(\theta;\eta)
    =
    \Theta_{\rm fine}(\theta;\eta)
    -
    \Theta_{\rm coarse}(\theta;\eta).
\end{equation}
Expanding the two fine steps gives
\begin{equation}
    d_{\rm TR}(\theta;\eta)
    =
    \frac{\eta^2}{4}J_U(\theta)U(\theta)
    +O(\eta^3).
\end{equation}
Since the leading-order time-renormalization force is again
\begin{equation}
    \nabla\phi_{\rm TR}(\theta)
    =
    \frac{\eta}{2}J_U(\theta)U(\theta)
    +O(\eta^2),
\end{equation}
we define the operational estimator
\begin{equation}
    F_{\rm TR}(\theta;\eta)
    =
    \frac{2}{\eta}
    \left[
    \Theta_{\rm fine}(\theta;\eta)
    -
    \Theta_{\rm coarse}(\theta;\eta)
    \right].
    \label{eq:app_F_TR}
\end{equation}
The scalar potential is reconstructed as
\begin{equation}
    \widehat{\phi}_{\rm TR}(\theta)
    =
    \int_0^1
    F_{\rm TR}(s\theta;\eta)\cdot \theta\,ds.
    \label{eq:app_phi_TR_measured}
\end{equation}

\subsection{Microscopic time-asymmetry potential \texorpdfstring{\(\phi_{\rm TA}\)}{phiTA}}

The microscopic time-asymmetry potential is measured by a forward--backward round trip. Starting at \(\theta_t\), take one forward step and then one backward step with the sign of the step size reversed:
\begin{equation}
    \theta_{t+1}=\theta_t-\eta U(\theta_t),
    \qquad
    \tilde{\theta}_t=\theta_t-\eta U(\theta_t)+\eta U(\theta_{t+1}).
    \label{eq:app_TA_round_trip_def}
\end{equation}
If the dynamics were microscopically reversible at this step size, then \(\tilde{\theta}_t=\theta_t\). The forward--backward residual is
\begin{equation}
    r_{\rm TA}(\theta_t;\eta)
    :=
    \theta_t-\tilde{\theta}_t.
    \label{eq:app_TA_residual}
\end{equation}
Following the force-level definition appearing in the main text, the measured time-asymmetry force is
\begin{equation}
    F_{\rm TA}(\theta_t;\eta)
    :=
    \frac{1}{2\eta}
    r_{\rm TA}(\theta_t;\eta)
    =
    \frac{\theta_t-\tilde{\theta}_t}{2\eta}.
    \label{eq:app_F_TA}
\end{equation}
Equivalently, this force-level convention defines the local scalar potential by
\begin{equation}
    2\eta \nabla\phi_{\rm TA}(\theta_t)
    :=
    \theta_t-\tilde{\theta}_t.
    \label{eq:app_TA_main_definition}
\end{equation}
Assume that \(U\) is \(C^3\) in a neighborhood of \(\theta_t\), with bounded derivatives there. Writing \(U_t:=U(\theta_t)\) and \(J_t:=J_U(\theta_t)\), Taylor expansion gives
\begin{equation}
\begin{aligned}
    U(\theta_{t+1})
    &=
    U(\theta_t-\eta U_t) \\
    &=
    U_t-\eta J_tU_t +O(\eta^2).
\end{aligned}
    \label{eq:app_TA_taylor_integral}
\end{equation}
Using the definition of \(\tilde{\theta}_t\), the forward--backward residual is
\begin{equation}
\begin{aligned}
    \theta_t-\tilde{\theta}_t
    &=
    \theta_t-
    \bigl(
    \theta_t-\eta U_t+\eta U(\theta_{t+1})
    \bigr) \\
    &=
    \eta\bigl(U_t-U(\theta_{t+1})\bigr) =
    \eta^2J_tU_t
    +O(\eta^3).
\end{aligned}
    \label{eq:app_TA_residual_expansion_integral}
\end{equation}
Therefore the force estimator satisfies
\begin{equation}
    F_{\rm TA}(\theta_t;\eta)
    =
    \frac{\eta}{2}J_tU_t
    +O(\eta^2).
\end{equation}
Under the symmetric-Jacobian assumption \(J_U=J_U^T\),
\begin{equation}
    \nabla\left(\frac{1}{2}\|U(\theta)\|^2\right)
    =
    J_U(\theta)^T U(\theta)
    =
    J_U(\theta)U(\theta).
\end{equation}
Thus the leading part of \(F_{\rm TA}\) is conservative. The scalar potential is reconstructed as
\begin{equation}
    \widehat{\phi}_{\rm TA}(\theta)
    =
    \int_0^1
    F_{\rm TA}(s\theta;\eta)\cdot\theta\,ds.
    \label{eq:app_phi_TA_measured}
\end{equation}
Consequently, the measured potential has the same leading term as the force-defined potential:
\begin{align}
    \widehat{\phi}_{\rm TA}(\theta)
    =&
    \phi_{\rm TA}(\theta)+O(\eta^2) \nonumber\\ 
    =&
    \frac{\eta}{4}\|U(\theta)\|^2
    +O(\eta^2)  =
    \phi_{\rm DE}(\theta)+O(\eta^2).
    \label{eq:app_TA_potential_leading}
\end{align}

\subsection{Stochastic-thermodynamic potential \texorpdfstring{\(\phi_{\rm ST}\)}{phiST}}

The stochastic-thermodynamic potential is measured using a regularized trajectory-probability ratio. We introduce a virtual Gaussian transition kernel
\begin{equation}
    p_\sigma(\theta'|\theta)
    \propto
    \exp\left[
    -\frac{
    \|\theta'-\theta+\eta U(\theta)\|^2
    }{2\sigma^2}
    \right],
    \label{eq:app_virtual_kernel}
\end{equation}
where \(\sigma^2\) is a small virtual noise variance. Given the deterministic update
\begin{equation}
    \theta_+
    =
    \theta-\eta U(\theta),
\end{equation}
we estimate the one-step bath entropy production by
\begin{equation}
    \Delta S_{\rm bath}(\theta;\eta,\sigma)
    =
    \log
    \frac{
    p_\sigma(\theta_+|\theta)
    }{
    p_\sigma(\theta|\theta_+)
    }.
\end{equation}
Using the Gaussian form above, this becomes
\begin{equation}
    \Delta S_{\rm bath}
    =
    \frac{1}{2\sigma^2}
    \left[
    \|\theta-\theta_+ + \eta U(\theta_+)\|^2
    -
    \|\theta_+-\theta + \eta U(\theta)\|^2
    \right].
    \label{eq:app_delta_S_bath}
\end{equation}
The scaled stochastic-thermodynamic potential is then measured as
\begin{equation}
    \widehat{\phi}_{\rm ST}(\theta)
    =
    \frac{\sigma^2}{8\eta}
    \Delta S_{\rm bath}(\theta;\eta,\sigma).
    \label{eq:app_phi_ST_measured}
\end{equation}
For the deterministic update, the forward residual
\[
\theta_+-\theta+\eta U(\theta)
\]
vanishes exactly, and the remaining backward residual measures the irreversibility of the transition under the virtual-noise regularization. In the small-\(\eta\) limit, this estimator satisfies
\begin{equation}
    \widehat{\phi}_{\rm ST}(\theta)
    =
    \frac{\eta}{4}\|U(\theta)\|^2
    +O(\eta^2).
\end{equation}

\subsection{Expected leading-order agreement}

The operational measurements above are designed so that, under the smoothness and symmetric-Jacobian assumptions,
\begin{equation}
    \widehat{\phi}_{\rm DE}(\theta)
    \approx
    \widehat{\phi}_{\rm TR}(\theta)
    \approx
    \widehat{\phi}_{\rm TA}(\theta)
    \approx
    \widehat{\phi}_{\rm ST}(\theta)
    \approx
    \frac{\eta}{4}\|U(\theta)\|^2.
    \label{eq:app_four_phi_agreement}
\end{equation}
Here \(\widehat{\phi}_{\rm TA}\) denotes the normalized quantity defined in Eq.~\eqref{eq:app_phi_TA_measured}. Therefore, sweeping over \(\eta\) at a fixed point \(\theta\) should reveal an approximately linear scaling in \(\eta\) for all four measured scalar potentials.

\subsection{Quadratic Test Problem}
\label{app:quadratic_test_problem}

We now describe the specific test problem used to validate the operational measurements. We consider a quadratic potential
\begin{equation}
    E(\theta)
    =
    \frac{1}{2}\theta^\top A\theta,
    \label{eq:app_quadratic_energy}
\end{equation}
where \(A\in\mathbb{R}^{d\times d}\) is positive definite. The update vector field is
\begin{equation}
    U(\theta)
    =
    \nabla E(\theta)
    =
    A\theta.
    \label{eq:app_quadratic_U}
\end{equation}
The discrete dynamics are therefore
\begin{equation}
    \theta_{k+1}
    =
    \theta_k-\eta A\theta_k.
    \label{eq:app_quadratic_discrete}
\end{equation}
This problem is useful because \(J_U=A\) is constant and symmetric, so all assumptions of the leading-order theory are exactly satisfied.

For this quadratic system, the leading-order reference potential is
\begin{equation}
    \phi_{\rm ref}(\theta)
    =
    \frac{\eta}{4}\|A\theta\|^2.
    \label{eq:app_quadratic_reference}
\end{equation}
All four operational measurements should agree with Eq.~\eqref{eq:app_quadratic_reference} to leading order:
\begin{equation}
    \widehat{\phi}_{\rm DE}(\theta)
    \approx
    \widehat{\phi}_{\rm TR}(\theta)
    \approx
    \widehat{\phi}_{\rm TA}(\theta)
    \approx
    \widehat{\phi}_{\rm ST}(\theta)
    \approx
    \frac{\eta}{4}\|A\theta\|^2.
    \label{eq:app_quadratic_expected_agreement}
\end{equation}

In the experiment, we fix a randomly sampled point \(\theta\) and evaluate all four operational estimators for a range of step sizes \(\eta\). The continuous-time flow needed for \(\phi_{\rm DE}\) is computed by a fourth-order Runge--Kutta integrator applied to
\begin{equation}
    \frac{d\vartheta}{d\tau}
    =
    -\eta A\vartheta,
    \qquad
    \tau\in[0,1].
\end{equation}
For each value of \(\eta\), the scalar potentials \(\widehat{\phi}_{\rm DE}\), \(\widehat{\phi}_{\rm TR}\), and \(\widehat{\phi}_{\rm TA}\) are reconstructed by straight-line integration from \(0\) to \(\theta\). The stochastic-thermodynamic quantity \(\widehat{\phi}_{\rm ST}\) is computed directly from the virtual-noise transition-ratio estimator in Eq.~\eqref{eq:app_phi_ST_measured}.

The expected outcome is that the four measured scalar potentials collapse onto the same curve when plotted against \(\eta\), with slope approximately one on a log-log plot:
\begin{equation}
    \widehat{\phi}_{\rm DE},
    \widehat{\phi}_{\rm TR},
    \widehat{\phi}_{\rm TA},
    \widehat{\phi}_{\rm ST}
    \propto
    \eta.
\end{equation}
This provides a direct operational check that the four distinct notions of irreversibility share the same leading-order potential in a setting where the theoretical assumptions are exactly controlled.

\section{Anomalous Fluctuations Experiment}
\subsection{Transformer}
\paragraph{Transformer experiment: tracking the (normalized) Adam update energy under learning-rate schedules.}
We analyze how the magnitude of Adam parameter updates evolves during training for a small decoder-only Transformer on an algorithmic task. 
The model is a $2$-layer causal Transformer (GPT-style) with $d_{\text{model}}=128$, $n_{\text{head}}=4$ attention heads, and a feedforward dimension of $512$. 
Token embeddings and learned positional embeddings are used, followed by a linear output head to logits over a vocabulary of size $V=64$.

\paragraph{Algorithmic task and supervision.}
We train on an autoregressive \emph{sequence reversal} task. 
Each example samples an i.i.d.\ token sequence $x=(x_1,\dots,x_T)$ of length $T=16$ uniformly from $\{4,\dots,V-1\}$. 
The target sequence is $y=(x_T,\dots,x_1)$. 
We concatenate input and target into a single stream
\[
[\texttt{BOS}],\, x_1,\dots,x_T,\, [\texttt{SEP}],\, y_1,\dots,y_T,\, [\texttt{EOS}],
\]
and train by next-token prediction with cross-entropy loss \emph{masked to only the target segment} (i.e.\ positions after \texttt{SEP}, including \texttt{EOS}). 
To induce label stochasticity, we independently corrupt each target token $y_t$ with probability $p=0.05$, replacing it with a uniformly sampled token from $\{4,\dots,V-1\}$.

\paragraph{Optimization, learning-rate schedules, and update-energy metric.}
We optimize all parameters using Adam with $(\beta_1,\beta_2)=(0.9,0.999)$, $\epsilon=10^{-8}$, minibatch size $B=64$, and base learning rate $\eta_0=10^{-3}$ for $S=2000$ steps. 
We compare three stepwise learning-rate schedules updated every $\Delta=1000$ steps by multiplying the current learning rate by a factor $k\in\{5,\,1,\,0.2\}$ (increasing, fixed, decreasing). 
At each step $s$, we measure the \emph{actual applied} Adam update by snapshotting parameters immediately before and after the optimizer step: $\Delta\theta^{(s)}=\theta^{(s)}-\theta^{(s+1)}$. 
We record the total squared update energy $U_s=\|\Delta\theta^{(s)}\|_2^2$ (summed over all parameter tensors) and report the learning-rate-normalized quantity $\widetilde{U}_s=U_s/\eta_s$, where $\eta_s$ is the current learning rate under the schedule.

\paragraph{Repetitions and visualization.}
For each schedule, we perform $R=3$ runs with different random seeds and store the full $\{\widetilde{U}_s\}_{s=1}^S$ trajectory from each run as a NumPy array, which is then loaded for plotting. 
We visualize the mean trajectory with a shaded $\pm1$ standard-deviation band across seeds, using a logarithmic $y$-axis and plotting markers on a subsampled grid (every 50 steps) for readability.

\subsection{RNN}
\paragraph{RNN experiment: tracking the (normalized) Adam update energy under learning-rate schedules.}
We study how the magnitude of parameter updates produced by Adam evolves during training for a recurrent sequence model. 
Our model is a gated recurrent unit (GRU) language model with $L=2$ recurrent layers and hidden size $h=256$. Inputs are embedded into $\mathbb{R}^{d}$ with $d=128$, and a linear readout maps hidden states to logits over a vocabulary of size $V=64$.

\paragraph{Algorithmic task and supervision.}
We train on a simple algorithmic \emph{sequence reversal} task formulated as next-token prediction. 
For each training example, we sample a token sequence $x=(x_1,\dots,x_T)$ of length $T=16$ i.i.d.\ uniformly from $\{4,\dots,V-1\}$. 
The desired output is the reversed sequence $y=(x_T,\dots,x_1)$. 
We concatenate input and target into a single autoregressive stream
\[
[\texttt{BOS}],\, x_1,\dots,x_T,\, [\texttt{SEP}],\, y_1,\dots,y_T,\, [\texttt{EOS}],
\]
and train the model to predict the next token at every position, \emph{but we mask the loss to only the target segment} (i.e.\ tokens after \texttt{SEP}, including \texttt{EOS}). 
To ensure nontrivial stochasticity, we introduce label noise by independently corrupting each target token $y_t$ with probability $p=0.05$, replacing it with a uniformly sampled token from $\{4,\dots,V-1\}$.

\paragraph{Optimization and learning-rate schedules.}
We optimize all parameters using Adam with $(\beta_1,\beta_2)=(0.9,0.999)$, $\epsilon=10^{-8}$, minibatch size $B=64$, and base learning rate $\eta_0=10^{-3}$. 
Training runs for $S=2000$ steps. 
We compare three stepwise learning-rate schedules updated every $\Delta=1000$ steps:
\[
\eta_{s} \leftarrow k\,\eta_{s}\quad \text{for } s\in\{\Delta,2\Delta,3\Delta,\dots\},
\]
with multiplicative factor $k\in\{5,\;1,\;0.2\}$ corresponding to \emph{increasing}, \emph{fixed}, and \emph{decreasing} learning-rate conditions.

\paragraph{Update-energy metric.}
At each step $s$, we measure the \emph{actual parameter change} applied by Adam. 
Let $\theta^{(s)}$ denote the parameter vector immediately \emph{before} the optimizer step and $\theta^{(s+1)}$ denote the parameter vector immediately \emph{after} the step. 
The applied update is $\Delta\theta^{(s)}=\theta^{(s)}-\theta^{(s+1)}$. 
We record the total squared update energy
\[
U_s \;=\; \|\Delta\theta^{(s)}\|_2^2 \;=\; \sum_i \|\Delta\theta^{(s)}_i\|_2^2,
\]
summing over all parameter tensors. 
To control for the learning-rate scale, we report the normalized quantity
\[
\widetilde{U}_s \;=\; \frac{U_s}{\eta_s},
\]
where $\eta_s$ is the current learning rate after applying the schedule at step $s$.

\paragraph{Repetitions and visualization.}
For each schedule $k\in\{5,1,0.2\}$, we run $R=3$ independent trials with different random seeds (affecting initialization and minibatch sampling). 
We store the full trajectory $\{\widetilde{U}_s\}_{s=1}^S$ from each run as a NumPy array and read these files back for analysis. 
Plots show the mean trajectory with a shaded band indicating $\pm 1$ standard deviation across seeds; markers are displayed on a subsampled grid (every 50 steps), and the vertical axis is shown on a logarithmic scale.

\subsection{Perceptron}

\paragraph{Perceptron (linear regression) experiment: tracking the (normalized) Adam update energy under learning-rate schedules.}
To provide a convex baseline, we repeat the same update-tracking procedure on a single-layer perceptron trained by mean-squared error (i.e.\ linear regression). 
The model is $f_\theta(x)=w^\top x$ with parameters $w\in\mathbb{R}^{d}$ (no bias), where $d=256$.
We generate a fixed synthetic dataset $\{(x_i,y_i)\}_{i=1}^{N}$ with $N=20000$ by sampling $x_i\sim\mathcal{N}(0,I_d)$ and setting
\[
y_i \;=\; w_\star^\top x_i \;+\; \sigma\,\varepsilon_i,
\qquad 
\varepsilon_i\sim\mathcal{N}(0,1),
\]
with $w_\star\sim\mathcal{N}(0,I_d)/\sqrt{d}$ and label noise $\sigma=0.1$. 
At each optimization step we draw a minibatch of size $B=256$ uniformly from the dataset and minimize the squared loss $\ell(\theta)=\|f_\theta(x)-y\|_2^2$.

\paragraph{Optimization, schedules, and metric.}
We optimize $w$ using Adam with $(\beta_1,\beta_2)=(0.9,0.999)$, $\epsilon=10^{-8}$, and base learning rate $\eta_0=10^{-3}$ for $S=2000$ steps. 
As in the sequence-model experiments, we compare three multiplicative learning-rate schedules updated every $\Delta=1000$ steps with factors $k\in\{5,\,1,\,0.2\}$. 
At each step $s$, we record the applied Adam update $\Delta w^{(s)}=w^{(s)}-w^{(s+1)}$ and compute the total squared update energy $U_s=\|\Delta w^{(s)}\|_2^2$. 
We report the learning-rate-normalized quantity $\widetilde{U}_s=U_s/\eta_s$ to facilitate comparisons across schedules.

\paragraph{Repetitions and visualization.}
For each schedule, we perform $R=3$ runs with different random seeds (affecting initialization and minibatch sampling, and the synthetic dataset generation) and store each $\{\widetilde{U}_s\}_{s=1}^S$ trajectory as a NumPy array, which is subsequently loaded for plotting. 
We display the mean trajectory with a shaded $\pm 1$ standard-deviation band across seeds on a logarithmic $y$-axis, using subsampled markers for readability.

\makeatletter
\@ifundefined{SupplementalMaterialImported}{%
\bibliographystyle{apsrev4-2}
\bibliography{ref}

@misc{li2018stochasticmodifiedequationsdynamics,
      title={Stochastic Modified Equations and Dynamics of Stochastic Gradient Algorithms I: Mathematical Foundations}, 
      author={Qianxiao Li and Cheng Tai and Weinan E},
      year={2018},
      eprint={1811.01558},
      archivePrefix={arXiv},
      primaryClass={cs.LG},
      url={https://arxiv.org/abs/1811.01558}, 
}

@article{hairer2006geometric,
  title={Geometric numerical integration},
  author={Hairer, Ernst and Hochbruck, Marlis and Iserles, Arieh and Lubich, Christian},
  journal={Oberwolfach Reports},
  volume={3},
  number={1},
  pages={805--882},
  year={2006}
}

@article{wilson1971renormalization,
  title={Renormalization group and critical phenomena. I. Renormalization group and the Kadanoff scaling picture},
  author={Wilson, Kenneth G},
  journal={Physical review B},
  volume={4},
  number={9},
  pages={3174},
  year={1971},
  publisher={APS}
}

@article{ziyin2025neural,
  title={Neural thermodynamics: Entropic forces in deep and universal representation learning},
  author={Ziyin, Liu and Xu, Yizhou and Chuang, Isaac},
  journal={NeurIPS},
  year={2025}
}

@article{murashita2014nonequilibrium,
  title={Nonequilibrium equalities in absolutely irreversible processes},
  author={Murashita, Y{\^u}to and Funo, Ken and Ueda, Masahito},
  journal={Physical Review E},
  volume={90},
  number={4},
  pages={042110},
  year={2014},
  publisher={APS}
}

@article{ziyin2023universal,
  title={Universal thermodynamic uncertainty relation in nonequilibrium dynamics},
  author={Ziyin, Liu and Ueda, Masahito},
  journal={Physical Review Research},
  volume={5},
  number={1},
  pages={013039},
  year={2023},
  publisher={APS}
}

@article{o2022time,
  title={Time irreversibility in active matter, from micro to macro},
  author={O’Byrne, J{\'e}r{\'e}my and Kafri, Yariv and Tailleur, Julien and van Wijland, Fr{\'e}d{\'e}ric},
  journal={Nature Reviews Physics},
  volume={4},
  number={3},
  pages={167--183},
  year={2022},
  publisher={Nature Publishing Group UK London}
}

@article{xu2026does,
  title={Does SGD Seek Flatness or Sharpness? An Exactly Solvable Model},
  author={Xu, Yizhou and Beneventano, Pierfrancesco and Chuang, Isaac and Ziyin, Liu},
  journal={arXiv preprint arXiv:2602.05065},
  year={2026}
}

@article{poggio2004general,
  title={General conditions for predictivity in learning theory},
  author={Poggio, Tomaso and Rifkin, Ryan and Mukherjee, Sayan and Niyogi, Partha},
  journal={Nature},
  volume={428},
  number={6981},
  pages={419--422},
  year={2004},
  publisher={Nature Publishing Group UK London}
}

@article{coppola2025renormalization,
  title={Renormalization group for deep neural networks: Universality of learning and scaling laws},
  author={Coppola, Gorka Peraza and Helias, Moritz and Ringel, Zohar},
  journal={arXiv preprint arXiv:2510.25553},
  year={2025}
}

@incollection{prigogine1973theory,
  title={Theory of dissipative structures},
  author={Prigogine, Ilya and Lefever, Ren{\'e}},
  booktitle={Synergetics: Cooperative phenomena in multi-component systems},
  pages={124--135},
  year={1973},
  publisher={Springer}
}

@article{halverson2021neural,
  title={Neural networks and quantum field theory},
  author={Halverson, James and Maiti, Anindita and Stoner, Keegan},
  journal={Machine Learning: Science and Technology},
  volume={2},
  number={3},
  pages={035002},
  year={2021},
  publisher={IOP Publishing}
}

@article{rotskoff2022trainability,
  title={Trainability and accuracy of artificial neural networks: An interacting particle system approach},
  author={Rotskoff, Grant and Vanden-Eijnden, Eric},
  journal={Communications on Pure and Applied Mathematics},
  volume={75},
  number={9},
  pages={1889--1935},
  year={2022},
  publisher={Wiley Online Library}
}

@article{ziyin2025parameter,
  title={Parameter Symmetry Potentially Unifies Deep Learning Theory},
  author={Ziyin, Liu and Xu, Yizhou and Poggio, Tomaso and Chuang, Isaac},
  journal={arXiv preprint arXiv:2502.05300},
  year={2025}
}

@article{xie2020diffusion,
  title={A diffusion theory for deep learning dynamics: Stochastic gradient descent exponentially favors flat minima},
  author={Xie, Zeke and Sato, Issei and Sugiyama, Masashi},
  journal={arXiv preprint arXiv:2002.03495},
  year={2020}
}

@inproceedings{liu2021noise,
  title={Noise and fluctuation of finite learning rate stochastic gradient descent},
  author={Liu, Kangqiao and Ziyin, Liu and Ueda, Masahito},
  booktitle={International Conference on Machine Learning},
  pages={7045--7056},
  year={2021},
  organization={PMLR}
}

@article{ziyin2025noise,
  title={Noise balance and stationary distribution of stochastic gradient descent},
  author={Ziyin, Liu and Li, Hongchao and Ueda, Masahito},
  journal={Physical Review E},
  volume={111},
  number={6},
  pages={065303},
  year={2025},
  publisher={APS}
}

@article{may1976simple,
  title={Simple mathematical models with very complicated dynamics},
  author={May, Robert M},
  journal={Nature},
  volume={261},
  number={5560},
  pages={459--467},
  year={1976},
  publisher={Springer}
}

@article{liu2019thermodynamic,
	title={Thermodynamic Uncertainty Relation for Arbitrary Initial States},
	author={Liu, Kangqiao and Gong, Zongping and Ueda, Masahito},
	journal={arXiv preprint arXiv:1912.11797},
	year={2019}
}

@article{goldt2017stochastic,
  title={Stochastic thermodynamics of learning},
  author={Goldt, Sebastian and Seifert, Udo},
  journal={Physical review letters},
  volume={118},
  number={1},
  pages={010601},
  year={2017},
  publisher={APS}
}

@article{mei2019mean,
  title={Mean-field theory of two-layers neural networks: dimension-free bounds and kernel limit},
  author={Mei, Song and Misiakiewicz, Theodor and Montanari, Andrea},
  journal={arXiv preprint arXiv:1902.06015},
  year={2019}
}

@article{seifert2008stochastic,
  title={Stochastic thermodynamics: principles and perspectives},
  author={Seifert, Udo},
  journal={The European Physical Journal B},
  volume={64},
  number={3-4},
  pages={423--431},
  year={2008},
  publisher={Springer}
}

@article{barrett2020implicit,
  title={Implicit gradient regularization},
  author={Barrett, David GT and Dherin, Benoit},
  journal={arXiv preprint arXiv:2009.11162},
  year={2020}
}

@article{smith2021origin,
  title={On the origin of implicit regularization in stochastic gradient descent},
  author={Smith, Samuel L and Dherin, Benoit and Barrett, David GT and De, Soham},
  journal={arXiv preprint arXiv:2101.12176},
  year={2021}
}

@article{journals/corr/KingmaB14_adam,
  author = {Kingma, Diederik P. and Ba, Jimmy},
  ee = {http://arxiv.org/abs/1412.6980},
  journal = {CoRR},
  title = {Adam: A Method for Stochastic Optimization.},
  url = {http://dblp.uni-trier.de/db/journals/corr/corr1412.html#KingmaB14},
  volume = {abs/1412.6980},
  year = 2014
}

@misc{Tieleman2012_rmsprop,
  title={{Lecture 6.5---RmsProp: Divide the gradient by a running average of its recent magnitude}},
  author={Tieleman, T. and Hinton, G.},
  howpublished={COURSERA: Neural Networks for Machine Learning},
  year={2012}
}

@article{yaida2018fluctuation,
  title={Fluctuation-dissipation relations for stochastic gradient descent},
  author={Yaida, Sho},
  journal={arXiv preprint arXiv:1810.00004},
  year={2018}
}

% Option 2: Manual bibliography (if you must)
% \begin{thebibliography}{99}
% \bibitem{key} Author, \textit{Journal} \textbf{vol}, page (year).
% \end{thebibliography}

\end{document}